%% file: glw.tex
\documentclass{amsart}

\usepackage{graphicx}
\usepackage{verbatim}

\newcommand{\M}{\mathcal{M}_{(\Lambda,m)}}
\newcommand{\m}{m_\Lambda}

\newcommand{\rH}{r_{\mathcal{H}}}
\newcommand{\rc}{r_{\mathcal{C}}}

\newcommand{\rcb}{\overline{\rc}}
\newcommand{\vrcb}{\lvert\rcb\rvert}
\newcommand{\onc}{\bigr\rvert_{\mathcal{C}^+}}
\newcommand{\kC}{\kappa_{\mathcal{C}}}

\newcommand{\ic}{\iota_\mathcal{C}}
\newcommand{\VC}{\Bigr\rvert_{\mathcal{C}^+}}
\newcommand{\Mp}{{M^\prime}}

\newcommand{\ud}{\mathrm{d}}
\newcommand{\g}{\stackrel{\circ}{\gamma}}
\newcommand{\dm}[1]{\ud\mu_{#1}}
\newcommand{\gb}[1]{\overline{g}_{#1}}
\newcommand{\cg}{\stackrel{\circ}{g}}

\newcommand{\nablab}{\nabla\!\!\!\!\!/\:}
\newcommand{\nablabc}{\stackrel{\circ}{\nablab}}
\newcommand{\nablac}{\stackrel{\circ}{\nabla}}
\newcommand{\db}{\ud\!\!\!/}
\newcommand{\Lb}{\mathcal{L}\!\!\!/\:}
\newcommand{\Oi}[1]{\Omega_{(#1)}}

\newcommand{\pd}[1]{\frac{\partial}{\partial #1}}
\newcommand{\dd}[2]{\frac{\partial #2}{\partial #1}}
\newcommand{\dr}{\frac{\partial}{\partial r}}

\newcommand{\sq}[1]{\bigl(#1\bigr)^2}
\newcommand{\sqv}[1]{\bigl\lvert #1 \bigr\rvert^2}
\newcommand{\deformt}[1]{{}^{(#1)}\pi}

\newcommand{\OdO}{\mathcal{O}\bigl(\frac{1}{r^{\delta}}\bigr)}
\newcommand{\Od}[1]{\mathcal{O}\bigl(\frac{1}{r^{#1+\delta}}\bigr)}
\newcommand{\pirm}{({}^\prime\pi_m)}

\numberwithin{equation}{section}
\DeclareMathOperator{\tr}{tr}

\theoremstyle{plain}
\newtheorem{prop}{Proposition}[section]
\newtheorem{thm}[prop]{Theorem}
\newtheorem{cor}[prop]{Corollary}
\newtheorem{lemma}[prop]{Lemma}

\theoremstyle{definition}
\newtheorem{defn}[prop]{Definition}

\theoremstyle{remark}
\newtheorem*{nota}{Notation}
\newtheorem{rmk}[prop]{Remark}

\begin{document}

\title[Linear waves on Kerr de Sitter]{Global results for linear waves on expanding Kerr and Schwarzschild de Sitter cosmologies}
\author{Volker Schlue}
\address{Department of Pure Mathematics and Mathematical Statistics\\ University of Cambridge\\ Wilberforce Road\\ Cambridge CB3 0WB\\ United Kingdom}
\curraddr{Mathematical Sciences Research Institute\\ 17 Gauss Way\\ Berkeley CA 94720\\ United States}
\email{vschlue@math.utoronto.ca}
\date{\today}
\thanks{The author would like to thank the UK Engineering and Physical Sciences Research Council, the Cambridge European Trust, and the European Research Council for their financial support.}

\begin{abstract}
In this global study of solutions to the linear wave equation on Schwarzschild de Sitter spacetimes we attend to the cosmological region of spacetime which is bounded in the past by cosmological horizons and to the future by a spacelike hypersurface at infinity.
We prove an energy estimate capturing the expansion of that region which combined with earlier results for the static region yields a global boundedness result for linear waves. 
It asserts that a general finite energy solution to the global initial value problem has a limit on the future boundary at infinity that can be viewed as a function on the standard cylinder with finite energy, and that moreover any decay along the cosmological horizon is inherited along the future boundary.
In particular, we exhibit an explicit nonvanishing quantity on the future boundary of the spacetime consistent with our expectations for the nonlinear stability problem.
Our results apply to a large class of expanding cosmologies near the Schwarzschild de Sitter geometry, in particular subextremal Kerr de Sitter spacetimes.
\end{abstract}

\maketitle

\tableofcontents

\section{Introduction}

The primary impetus for the study of the linear wave equation on black hole spacetimes is that some aspects of the stability or instability of dynamical black holes are already foreshadowed in the global behaviour of linear waves \cite{dr:bhslp,a:in}. 
In this paper, we affirm the global linear stability of a range of \emph{cosmological} black holes, namely a class of spacetimes sufficiently close to spherically symmetric subextremal vacuum black holes with positive cosmological constant. While some important aspects of this problem pertaining to the linear stability of the \emph{stationary} domains have already been resolved in \cite{bh:sds,dr:sds,mbv:sds,sd:beyond,vd:microlocal,dr:kerr:bounded}, we focus here on the \emph{cosmological} region to the future of the stationary domains undergoing accelerated expansion. We identify a \emph{global redshift effect} as the relevant stability mechanism which combined with previous work for the stationary region yields a global linear stability result.

\begin{figure}
\includegraphics{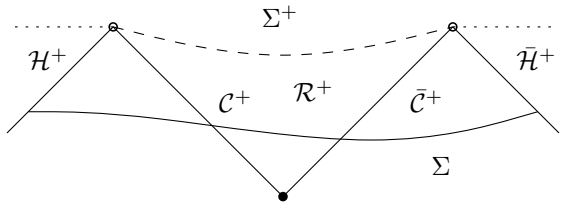}
\caption{Global geometry of Schwarzschild de Sitter. The expanding region is the domain bounded in the past by the cosmological horizons $\mathcal{C}^+$ and $\bar{\mathcal{C}^+}$ and to the future by the spacelike hypersurface $\Sigma^+$ at infinity. $\mathcal{H}^+$ and $\bar{\mathcal{H}}^+$ denote the event horizons of the adjacent black hole interiors.}
\label{fig:cauchyexpanding}
\end{figure}

The cosmological black hole spacetimes underlying the present study are precisely the members of the Schwarzschild de Sitter family in the subextremal range.
We may depict these spacetimes as an infinite chain of black hole interiors, \emph{static} black hole exteriors, and \emph{cosmological} regions separated by event horizons and cosmological horizons respectively \cite{gh:ceh}. 
The distinctive feature in comparison to the case of vanishing cosmological constant is here that future null infinity is a \emph{spacelike} hypersurface $\Sigma^+$, along with the existence of cosmological horizons $\mathcal{C}^+$, $\bar{\mathcal{C}}^+$ as future boundaries of the static regions in the outgoing direction. The domain bounded in the past by the cosmological horizons $\mathcal{C}^+\cup\bar{\mathcal{C}}^+$ and to the future by $\Sigma^+$ has the property that the area of each sphere is strictly increasing towards the future and shall thus be referred to as the \emph{expanding region} $\mathcal{R}^+$; see figure~\ref{fig:cauchyexpanding}.

The expanding region is foliated by spacelike hypersurfaces of constant area radius. We exhibit in this paper an energy of linear waves which is \emph{monotone} with respect to this foliation. Moreover, we show that any solution to the linear wave equation with initially finite energy on a Cauchy hypersurface $\Sigma$ lying in the past of $\Sigma^+$ (such that $\Sigma^+$ is in the domain of dependence of $\Sigma$, see figure \ref{fig:cauchyexpanding}) \emph{attains a limit on $\Sigma^+$ which can be viewed as a function on the standard cylinder $\mathbb{R}\times\mathbb{S}^2$ with finite energy}. In addition, it will be shown that as a consequence of the global redshift effect, and the local redshift effect on the horizon, polynomial or exponential decay along the cosmological horizon $\mathcal{C}^+$ will be inherited along the future boundary $\Sigma^+$. 

The results of this paper apply to solutions of the wave equation on a large class of spacetimes which preserve the global causal geometry and are sufficiently close to a Schwarzschild de Sitter spacetime \emph{due to the stability of the redshift effect}. In particular our results apply to \emph{subextremal Kerr de Sitter spacetimes}. In fact however, these spacetimes do not need to possess any symmetries in the expanding region. 
The ideas of the proof originate in the author's treatment of linear waves on de Sitter cosmologies \cite{vs:thesis} referred to in the epilogue of this paper.

While the resolution of the \emph{black hole stability conjecture} has remained elusive for many years \cite{d:bhsp}, the motivation for the linear theory developed in this paper is largely drawn from our expectation that the nonlinear stability problem is within reach in the context of \emph{expanding spacetimes}. 
The favorable role of a \emph{positive} cosmological constant, which allows for expanding cosmological solutions to the Einstein equations, was first recognized by Friedrich in his proof of the stability of the de Sitter solution \cite{f:ds}. While the conformal method has since then been extended to scalar field models \cite{k:fs}, the stabilizing effect of the cosmological constant may be seen most clearly manifested in the recent works of \cite{hr:s} on the scalar field and \cite{rs:siee,s:se} on perfect fluids, (see also \cite{kl:pureradiation} for the pure radiation case).

In view of these results for the nonlinear field equations with positive cosmological constant and the linear theory presented in this paper,
we expect the characteristic initial value problem for perturbations of Schwarzschild de Sitter data on $\mathcal{C}^+\cup\bar{\mathcal{C}}^+$ to be more tractable than the major open problem of establishing the nonlinear stability of the past of $\mathcal{H}^+\cup\mathcal{C}^+$; c.f.~figure~\ref{fig:cauchyexpanding}.

\section{Boundedness and decay results}
\label{sec:r}

\subsection{Geometry of the Schwarzschild de Sitter family}\label{sec:r:geometry}
The black hole spacetimes providing the starting point for this paper are the members of the Schwarzschild de Sitter family $(\M,g)$ for which $\Lambda>0$ and $0<3m\sqrt{\Lambda}<1$ (where $\Lambda$ denotes the cosmological constant and $m$ the mass of the black hole). 
An important role in our analysis and in the depiction of the causal geometry of this \emph{spherically symmetric} spacetime family is played by the function $r$ on $\mathcal{Q}=\mathcal{M}/\mathrm{SO(3)}$ which takes the value of the area radius on a sphere of symmetry. In other words, $r$ is constant on the orbits of the $\mathrm{SO}(3)$ subgroup of the isometry group of the spacetime $(\M,g)$.
In the given parameter range the spacetime $(\M,g)$ is partitioned into black/white hole interiors $\mathcal{B}$, static exteriors $\mathcal{S}$ and cosmological regions $\mathcal{R}$ in the sense that
\begin{equation}
  \label{eq:sds:topology}
  \mathcal{M}=\ldots\cup\overline{\mathcal{B}\cup\mathcal{S}\cup\mathcal{R}\cup\mathcal{S}\cup\mathcal{B}}\cup\ldots
\end{equation}
where a sphere is in $\mathcal{B}$, $\mathcal{S}$ or $\mathcal{R}$ according as to whether its area radius is in the interval $0<r<\rH$, $\rH<r<\rc$ or $r>\rc$ respectively; here $2m<\rH<3m<\rc$ are the positive roots of a cubic polynomial with coefficients in $\Lambda$ and $m$ (for explicit expressions see \cite{kr:effects}) which mark the \emph{event}, $\mathcal{H}$, $\bar{\mathcal{H}}$, and \emph{cosmological horizons}, $\mathcal{C}$, $\bar{\mathcal{C}}$, respectively. 

In the following we fix $\Lambda>0$, $0<3m<1/\sqrt{\Lambda}$ and pick a domain $\mathcal{R}$; as discussed in \cite{kr:effects,vs:thesis} --- and included for completeness in Section \ref{sec:sds:dn} --- the cosmological region $\mathcal{R}$ and its adjacent static domains $\mathcal{S}$ can be covered with a single \emph{double null coordinate system $(u,v)$} (i.e.~$u$, $v$ are functions on $\mathcal{Q}$ increasing towards the future such that the level sets of $u$ are outgoing, and the levels sets of $v$ are ingoing null lines) which identifies the cosmological horizons $\mathcal{C}$, $\bar{\mathcal{C}}$, with the null hypersurfaces $v=0$, $u=0$, respectively (see figure \ref{fig:dn}).
We refer to the component of $\mathcal{R}$ to the future of the static regions $\mathcal{S}$ as the \emph{expanding} region $\mathcal{R}^+$ which in the above chart corresponds to $0<uv<1$. The hyperbola $uv=1$ represents the future timelike boundary of $\mathcal{R}^+$: a spacelike cylinder composed of spheres of infinite radius. We shall also denote by $\mathcal{C}^+=\mathcal{C}\cap\{u\geq 0\}$, $\bar{\mathcal{C}}^+=\bar{\mathcal{C}}\cap\{v\geq 0\}$, the components of $\mathcal{C}$, $\bar{\mathcal{C}}$ to the future of $\mathcal{S}$, respectively, so that the past boundary of $\mathcal{R}^+$ is precisely $\mathcal{C}^+\cup\bar{\mathcal{C}}^+$.

\begin{figure}
  \includegraphics{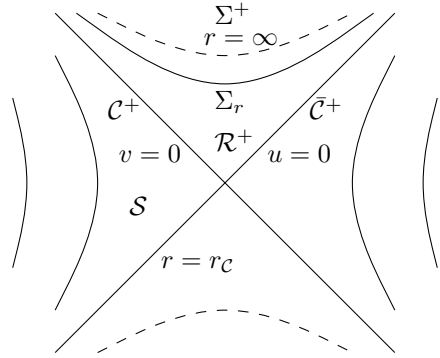}
  \caption{Causal geometry of the cosmological region $\mathcal{R}=\{0<uv<1\}$ and the adjacent static exteriors $\mathcal{S}=\{uv<0\}$. The expanding region $\mathcal{R}^+$ is the component of $\mathcal{R}$ to the future of the cosmological horizons $\mathcal{C}=\{v=0\}$ and $\bar{\mathcal{C}}=\{u=0\}$, and the future boundary $\Sigma^+$ of $\mathcal{R}^+$ is represented by the hyperbola $uv=1$.}
  \label{fig:dn}
\end{figure}

In the expanding region the area radius $r>\rc$ is strictly increasing along any future-directed timelike curve. We shall thus view $r$ as a \emph{time function} for the expanding region and obtain a foliation of $\mathcal{R}^+$ by the level sets $\Sigma_r$ of $r:\mathcal{R}^+\to(\rc,\infty)$ with the \emph{lapse function}
\begin{equation}
  \label{eq:sds:lapse}
  \phi\doteq\frac{1}{\sqrt{-g(V,V)}}=\frac{1}{\sqrt{\frac{\Lambda r^2}{3}+\frac{2m}{r}-1}}\,,
\end{equation}
where $V^\mu=-g^{\mu\nu}\partial_\nu r$ denotes the gradient vectorfield of $r$.
The metric $g$ decomposed in this way reads
\begin{equation}
  \label{eq:sds:metric:levelsets}
  g=-\phi^2\ud r^2+\gb{r}\,,
\end{equation}
where $\gb{r}$ is the induced metric on $\Sigma_r$; in fact we can choose a function $t:\mathcal{R}^+\to(-\infty,\infty)$ which is constant on the lines $u=\alpha v$, $\alpha\in(0,\infty)$, such that 
\begin{equation}
  \label{eq:sds:inducedmetric}
  \gb{r}=\Bigl(\frac{\Lambda r^2}{3}+\frac{2m}{r}-1\Bigr)\ud t^2+r^2\g\,,
\end{equation}
where $\g$ is the standard metric on the unit sphere.
The level sets $\Sigma_r$ are thus \emph{conformal to the standard cylinder} $(\mathbb{R}\times\mathbb{S}^2,\cg=\ud t^2+\g)$, and we note in particular for the volume form
\begin{equation}
  \label{eq:sds:overview:volumeform}
  \phi\:\dm{\gb{r}}=r^2\:\dm{\cg}\,.
\end{equation}
The normal to $\Sigma_r$ is given by $n=\phi\,V$, while we denote by
\begin{equation}
  \label{eq:sds:dr}
  \dr=\phi^2\:V=\phi\,n
\end{equation}
the vectorfield colinear to $V$ with the property that $\partial_r r=1$; (in the above $(t,r)$ coordinates $\partial_r$ is indeed the coordinate vectorfield).

Furthermore there exists a global Killing vectorfield $T$ which coincides with the stationary vectorfield in the static regions (where it is timelike), and with the coordinate vectorfield $\pd{t}$ in the expanding region (where it is spacelike); its integral curves in the static region correspond to physical time translations, and its properties on the cosmological horizons (where it is null) are intimately linked to the positivity of the surface gravity (see Section~\ref{sec:sds:dn}).
Note that $T$ and the generators of the spherical isometry $\Oi{i}:i=1,2,3$ (which are discussed in Appendix~\ref{a:coercivity}) are orthogonal to $n$; the tangent space to $\Sigma_r$ is thus spanned by \emph{Killing vectorfields} at each point.

\subsection{Linear wave equations}
This paper is devoted to the global study of solutions to the linear wave equation
\begin{equation}
  \label{eq:sds:wave}
  \Box_g\psi=0
\end{equation}
on the family of black hole spacetimes $(\M,g)$ described in Section \ref{sec:r:geometry}.
It is entirely based on the use of energy currents which enlarges its scope of applicability to wave equations on general perturbations of $(\M,g)$ that preserve the global causal geometry. This is demonstrated explicitly for a class of nearby cosmologies that include the entire range of \emph{rotating subextremal Kerr de Sitter spacetimes} in Section~\ref{sec:general}. 
In Section~\ref{sec:kg} we shall also explain how our currents can be modified to incorporate a mass term,
\begin{equation}
  \label{eq:sds:kg}
  \Box_g\psi=\m\psi\,,\qquad \m\geq 0\,.
\end{equation}

Our results apply to a class of \emph{finite energy} solutions for the Cauchy problem associated to \eqref{eq:sds:wave}.
Let us recall for this purpose briefly the standard energy current $J^X$ associated to the multiplier vectorfield $X$, a 1-form defined by
\begin{equation}
  \label{eq:r:ec}
  J^X[\psi]\cdot N=T(X,N)\,,
\end{equation}
where $T$ denotes the standard energy momentum tensor associated to \eqref{eq:sds:wave} (for an introductory discussion of Lagrangian field theories in a curved spacetime see \cite{ch:mpI}).
Given a spacelike hypersurface $(\Sigma,\gb{})$ with normal $n$ we can think of
\begin{equation}
  \label{eq:r:energy}
  \int_\Sigma J^n[\psi]\cdot n\,\dm{\gb{}}\geq 0
\end{equation}
as the total energy flux of the wave $\psi$ through $\Sigma$.
Our assumptions on the solutions of \eqref{eq:sds:wave} considered in this paper are stated in terms of a finiteness condition on the level of initial data for an expression of the form \eqref{eq:r:energy}.

\subsection{Integral estimates}\label{sec:r:integral}
Our results provide \emph{integral} bounds for the energy density of linear waves on spacelike noncompact hypersurfaces in the expanding region, in particular on the future boundary, in terms of the energy on general spacelike hypersurfaces, and thus emphasize the \emph{global} nature of our approach.

\subsubsection{Schwarzschild de Sitter spacetimes}
In Section \ref{sec:sds:globalredshift} we demonstrate the presence of a ``global redshift'' which leads us to our key result for the expanding region.
\begin{prop}\label{prop:r:expanding}
Let $\psi$ be a solution to \eqref{eq:sds:wave} then for all $r_2>r_1>\rc$ we have
\begin{multline*}
  \int_{\Sigma_{r_2}}\biggl\{\frac{1}{\phi^2}\sq{\dd{r}{\psi}}+\phi^2\sq{\dd{t}{\psi}}+\sqv{\nablab\psi}\biggr\}\phi\,\dm{\gb{r_2}}\leq\\
  \leq\int_{\Sigma_{r_1}}\biggl\{\frac{1}{\phi^2}\sq{\dd{r}{\psi}}+\phi^2\sq{\dd{t}{\psi}}+\sqv{\nablab\psi}\biggr\}\phi\,\dm{\gb{r_1}}\,.
\end{multline*}
\end{prop}
\begin{nota}
Here $\nablab$ denotes the projection of the covariant derivative $\nabla$ of $g$ to the sphere of area radius $r$ at this point, and we simply write $\lvert\cdot\rvert$ for the induced norm on this sphere.
\end{nota}

It is essentially due to the \emph{local} redshift effect of the cosmological horizons that Prop.~\ref{prop:r:expanding} can be turned into a global result.
We should regard the \emph{positivity of the surface gravity} of the cosmological horizons (see Section \ref{sec:sds:dn}) as the origin of this effect which is here manifestly exploited in two ways. 

Firstly in a construction of a redshift vectorfield $N$ (we refer to the epilogue in \cite{dr:c} for a general discussion) which gives rise to an energy identity on the spacetime domain bounded in the future by $\Sigma_{r_1}$ and in the past by $\mathcal{C}^+\cup\bar{\mathcal{C}}^+$; see Section \ref{sec:sds:localredshift}. By choosing $r_1>\rc$ sufficiently small, we obtain an estimate for the right hand side in Prop.~\ref{prop:r:expanding} in terms of the corresponding energy \emph{on the cosmological horizons}. In fact, we have:
\begin{prop}\label{prop:sds:local}
Let $\Sigma\subset\overline{\mathcal{R}^+}$ be a spacelike hypersurface with normal $n$ crossing the cosmological horizons to the future of $\mathcal{C}^+\cap\bar{\mathcal{C}}^+$, and let us denote by $\mathcal{C}_0^+=\mathcal{C}^+\cap J^+(\Sigma)$ the segment of $\mathcal{C}^+$ to the future of $\Sigma$.
There exists a \emph{timelike}  vectorfield $N$ (which is invariant under the 1-parameter group of isometries $\varphi_t$ generated by $T$) in a neighborhood of $\mathcal{C}_0^+$ of the form $\mathcal{N}=\{r_0^-\leq r\leq r_0^+\}\cap J^+(\Sigma)$ (where $r_0^-<\rc<r_0^+$ are constants that only depend on $\Sigma$, $\Lambda$ and $m$) with the property that for all solutions $\psi$ to the wave equation \eqref{eq:sds:wave},
\begin{multline*}
\int_{\Sigma_{r_1}\cap J^+(\Sigma)}\biggl\{\frac{1}{\phi^2}\sq{\dd{r}{\psi}}+\phi^2\sq{\dd{t}{\psi}}+\sqv{\nablab\psi}\biggr\}\phi\,\dm{\gb{r_1}}\leq\\
\leq C(r_1)\int_{\mathcal{C}_0^+}{}^\ast J^N[\psi]+C(r_1)\int_{\Sigma^\prime}J^N[\psi]\cdot n\:\dm{\gb{}}\,,
\end{multline*}
whenever $\rc<r_1<r_0^+$ (provided $r_1$ is chosen small enough so that $\Sigma_{r_1}\cap\Sigma\neq\emptyset$) where $C$ is a constant that only depends on $r_1$ and $\Sigma^\prime=J^-(\Sigma_{r_1})\cap\Sigma$ denotes the spacelike hypersurface $\Sigma$ truncated by $\Sigma_{r_1}$.
\end{prop}
\begin{nota}
Given a 1-form $J$ we denote its dual with respect to the volume form by ${}^\ast J$. The notation $\gb{}$ refers to the first fundamental form of $\Sigma$.
\end{nota}

\begin{figure}
\includegraphics{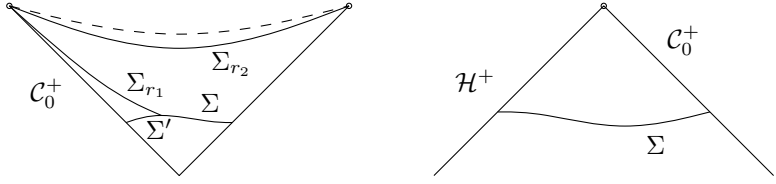}
\caption{Notation for hypersurfaces in Propositions~\ref{prop:r:expanding}-\ref{prop:r:static}.}
\label{fig:not:redshift:props}
\end{figure}

Secondly in the work of \cite{dr:sds} which is occupied with the construction of suitable currents for the \emph{static region}. We cite here a \emph{boundedness} result which yields control on the nondegenerate energy on the cosmological horizon in terms of a natural energy on a spacelike Cauchy hypersurface and thus closes our estimate as desired. 

\begin{prop}[Prop.~10.3.2, \cite{dr:sds}]
\label{prop:r:static}
Let $\Sigma\subset\overline{\mathcal{S}}$ be a spacelike hypersurface with normal $n$ in the static region crossing the horizons to the future of the bifurcation spheres $\mathcal{C}\cap\bar{\mathcal{C}}$, $\mathcal{H}\cap\bar{\mathcal{H}}$, and denote by $\mathcal{C}_0^+=\mathcal{C}^+\cap J^+(\Sigma)$. 
We have for all solutions $\psi$ to \eqref{eq:sds:wave},
\begin{equation*}
  \int_{\mathcal{C}_0^+}{}^\ast J^N[\psi]\leq C\:\int_{\Sigma}J^n[\psi]\cdot n\,\dm{\gb{}}\,,
\end{equation*}  
where $C$ is a constant that only depends on $\Sigma$, $\Lambda$ and $m$, whenever the right hand side is finite.
\end{prop}
\begin{rmk}
Here $N$ is precisely the vectorfield referred to in Prop.~\ref{prop:sds:local}. While \cite{dr:sds} relies on an extension of $N$ into the static region, we require for our purposes an extension  into the expanding region; however, on the cosmological horizon $\mathcal{C}^+$ the two constructions coincide.
\end{rmk}
\begin{rmk}
 We point out that in the literature the \emph{boundedness} statement of Prop.~\ref{prop:r:static} is only found in \cite{dr:sds}. The implications of related decay results for the static region \cite{dr:sds,bh:sds,mbv:sds,sd:quasi,sd:beyond} are discussed in Remark~\ref{rmk:cor:decay}.
\end{rmk}

Let us consider global solutions to the Cauchy problem \eqref{eq:sds:wave} with initial data prescribed on a spacelike hypersurface $\Sigma\subset\overline{\mathcal{S}\cup\mathcal{R}^+}$ as depicted in figure \ref{fig:cauchyexpanding}; more precisely we say:
\begin{defn}
A spacelike hypersurface $\Sigma\subset\M$ is called a \emph{$\Sigma^+$-Cauchy hypersurface} if the future boundary $\Sigma^+$ of a chosen expanding region $\mathcal{R}^+\subset\M$ is in the domain of dependence of $\Sigma\subset J^-(\Sigma^+)$, and $\Sigma\subset\overline{\mathcal{R}^+\cup\mathcal{S}}$ crosses the horizons $\mathcal{H}$, $\mathcal{C}$, $\bar{\mathcal{C}}$, and $\bar{\mathcal{H}}$ to the future of the bifurcation spheres $\bar{\mathcal{H}}\cap\mathcal{H}$, $\bar{\mathcal{C}}\cap\mathcal{C}$, and $\bar{\mathcal{H}}\cap\mathcal{H}$, respectively.
\end{defn}
As a consequence of the global and local redshift effect on and near the cosmological horizons --- namely Propositions~\ref{prop:r:expanding}, \ref{prop:sds:local} and \ref{prop:r:static} --- the energy on $\Sigma_{r}\cap J^+(\Sigma)$ is finite for any $r>\rc$ provided it is initally finite on $\Sigma$.
We may thus take the limit $r\to\infty$ in the global redshift Proposition~\ref{prop:r:expanding} which yields the finiteness of an explicit limiting quantity on $\Sigma^+$. In fact, we infer from \eqref{eq:sds:overview:volumeform} that
\begin{equation}
  \label{eq:sds:overview:finiteintegral}
  \lim_{r\to\infty}\int_{\Sigma_{r}}\biggl\{\phi^2\sq{\dd{t}{\psi}}+\sqv{\nablab\psi}\biggr\}\phi\,\dm{\gb{r}}=\int_{\Sigma^+}\Bigl\{\frac{3}{\Lambda}\sq{\dd{t}{\psi}}+\sqv{\nablabc\psi}\Bigr\}\dm{\cg}\,
\end{equation}
is the \emph{rescaled asymptotic energy} of a solution to the Cauchy problem viewed as a function on the standard cylinder; (here $\nablabc$ denotes the covariant derivative on the standard sphere $(\mathbb{S}^2,\g)$).

\begin{thm}\label{thm:sds}
Let $\Sigma^+$ be the future boundary of a chosen expanding region $\mathcal{R}^+\subset\M$ of Schwarzschild de Sitter (with $0<3m\sqrt{\Lambda}<1$, $\Lambda>0$) endowed with the standard metric $\cg$ of the cylinder $\mathbb{R}\times\mathbb{S}^2$, and let $\Sigma$ be a $\Sigma^+$-Cauchy hypersurface with normal $n$. Then all solutions $\psi$ to the wave equation \eqref{eq:sds:wave} with initial finite energy on $\Sigma$,
\begin{equation}
  \label{eq:thm:sds:condition}
  D[\psi]\doteq\int_\Sigma J^n[\psi]\cdot n\,\dm{\gb{}}<\infty\,,\qquad\psi\rvert_{\Sigma}\in\mathrm{H}^1(\Sigma)\,,
\end{equation}
are globally bounded on $\overline{\mathcal{R}^+}$ in the energy defined by Prop.~\ref{prop:r:expanding}, and have a limit on $\Sigma^+$ in $\stackrel{\circ}{\mathrm{H}}\!\!{}^1(\mathbb{R}\times\mathbb{S}^2)$. Moreover, the limit \emph{as a function on $\mathbb{R}\times\mathbb{S}^2$} satisfies
\begin{equation}\label{eq:thm:sds:energy}
  \int_{\Sigma^+}\sqv{\nablac\psi}\dm{\cg}\leq C(\Lambda,m,\Sigma)\,D[\psi]\,,
\end{equation}
where $C$ is a constant that only depends on $\Lambda$, $m$, and $\Sigma$, (and $\nablac$ denotes the standard gradient on the cylinder $\mathbb{R}\times\mathbb{S}^2$).
\end{thm}

\begin{rmk}
Our result Thm.~\ref{thm:sds} can be read as an implicit version of a decay statement for the ``natural'' energy in this problem. For we establish the boundedness of a limiting \emph{rescaled} quantity which corresponds to the decay of the induced geometric quantities arising in this setting; note e.g.~on the level of energy densities that
\begin{equation*}
  \bigl\lvert\nablabc\psi\bigr\rvert=r\bigl\lvert\nablab\psi\bigr\rvert\,,
\end{equation*}
i.e.~the angular derivatives on the rescaled standard sphere differ from the induced angular derivatives on a sphere on $\Sigma_r$ by a factor of $r$.

Moreover the result is in agreement with our expectation for the nonlinear stability problem. Although we expect to recover the same global \emph{causal} geometry, the dynamical development of a perturbation of Schwarzschild de Sitter initial data is not expected to settle down to the \emph{exact} geometry of the expanding region of a Schwarzschild de Sitter solution. This is captured by a nonvanishing bounded quantity on the future boundary.
 
Furthermore this behaviour persists in the larger class of perturbed spacetimes to be discussed in Section~\ref{sec:r:integral:gen}.

Note that we have dropped the normal derivative in the limit \eqref{eq:sds:overview:finiteintegral}. If we viewed the asymptotic quantity recovered in Thm.~\ref{thm:sds} as initial data for the corresponding backward problem, then the rescaled asymptotic normal derivative of Prop.~\ref{prop:r:expanding} would be part of the data at infinity.
\end{rmk}

This result can be viewed as the most general global energy bound in the class of finite energy solutions to \eqref{eq:sds:wave}. While in Thm.~\ref{thm:sds} we only require the initial energy to be finite, it is known that under suitable assumptions \emph{on the higher order energies} (i.e.~under stronger regularity assumptions on the initial data) the solutions to \eqref{eq:sds:wave} will in fact \emph{decay} along the cosmological horizons. It is an immediate consequence of our approach (see also Section~\ref{sec:sds:localisation}) that our estimates can be localised to show that any decay along the horizons is in fact inherited along the future boundary.

\begin{nota}
Let $\tau$ be a function on $\mathcal{C}^+$ which is constant on the spheres of symmetry such that $\tau=\tau_0>0$ on a chosen sphere $S\subset\mathcal{C}^+$ to the future of $\mathcal{C}^+\cap\bar{\mathcal{C}}^+$ and
\begin{equation*}
  T\cdot\tau\onc=1\,.
\end{equation*}
We denote by $\mathcal{C}_{\tau_1}^+$ the segment of $\mathcal{C}^+$ lying to the future of the sphere with value $\tau=\tau_1$, and by $\mathrm{C}_{\tau_1}$ the outgoing null hypersurface from the sphere $\tau=\tau_1$; see figure \ref{fig:not:segments}. Moreover let us denote by\footnote{We use standard notation for causal sets in Lorentzian geometry, see e.g.~\cite{he:gr}.}
\begin{equation*}
  \Sigma_\tau^+=\mathrm{J}^+(\mathcal{C}_\tau^+\cup\mathrm{C}_\tau)\cap\Sigma^+
\end{equation*}
the future boundary of the spacetime domain to the future of $\mathcal{C}_\tau^+\cup\mathrm{C}_\tau$.
\end{nota}

\begin{figure}
\includegraphics{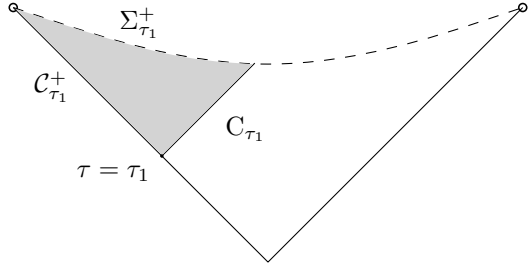}
\caption{Notation for segments of the cosmological horizon $\mathcal{C}^+$ and the future boundary $\Sigma^+$.}
\label{fig:not:segments}
\end{figure}

\begin{cor}\label{cor:decay}
Assume $\psi$ is a solution to \eqref{eq:sds:wave} which decays polynomially along $\mathcal{C}^+$ in the sense that for some fixed integer $k\in\mathbb{N}$ there exists a constant $C_k>0$ such that
\begin{equation}
  \label{eq:cor:decay:assumption}
  \int_{\mathcal{C}^+_\tau}{}^*J^N[\psi]\leq\frac{C_k}{\tau^k}\qquad(\tau>\tau_0)\,.
\end{equation}
Then there exists a constant $C_k\leq C<\infty$ such that also
\begin{equation}
  \label{eq:cor:decay:conclusion}
  \int_{\Sigma^+_\tau}\sqv{\nablac\psi}\dm{\cg}\leq \frac{C}{\tau^k}\qquad(\tau>\tau_0)\,.
\end{equation}
Moreover, if the decay along $\mathcal{C}^+$ is assumed to be exponential, i.e.~if for some $\beta>0$ there exists a constant $C_\beta>0$ such that
\begin{equation}
  \label{eq:cor:decay:assumption:exp}
  \int_{\mathcal{C}^+_\tau}{}^*J^N[\psi]\leq C_\beta\:e^{-\beta\tau}\qquad(\tau>\tau_0)\,,
\end{equation}
then also the decay along the future boundary $\Sigma^+$ is exponential, with a rate however that is not faster than dictated by the local redshift effect of the cosmological horizon; in fact then there exists a constant $C_\beta<C<\infty$ such that
\begin{equation}
  \label{eq:cor:decay:conclusion:exp}
  \int_{\Sigma^+_\tau}\sqv{\nablac\psi}\dm{\cg}\leq C\,e^{-\min\{b,\beta\}\:\tau}\qquad(\tau>\tau_0)\,,
\end{equation}
where $b$ is a constant that only depends on $\Lambda$, and $m$; (here for simplicity $\beta\neq b$).
\end{cor}

\begin{rmk}\label{rmk:cor:decay}
The assumptions of Cor.~\ref{cor:decay} are known to be satisfied under higher regularity conditions on the initial data \cite{dr:sds,sd:beyond}. (The constants $C_k$, $C_\beta$, then depend on higher order norms of the initial data.)

More precisely, it is shown in \cite{dr:sds} that \eqref{eq:cor:decay:assumption} is satisfied with $k=m+1$ (for any $m\in\mathbb{N}$) for any solution to \eqref{eq:sds:wave} in the domain of dependence of a $\Sigma^+$-Cauchy hypersurface $\Sigma$ with normal $n$ (where $\tau=\tau_0$ on $\Sigma\cap\mathcal{C}^+$) provided
\begin{equation}
  \label{eq:rmk:assumption}
  \int_\Sigma \sum_{\lvert\alpha\rvert\leq m}J^n[\Omega^\alpha\psi]\cdot n\:\dm{\gb{}}<\infty\,,
\end{equation}
where $\alpha$ is a multi-index of order $\vert\alpha\rvert=\alpha_1+\alpha_2+\alpha_3$ and $\Omega^\alpha=\Omega_1^{\alpha_1}\Omega_2^{\alpha_2}\Omega_3^{\alpha_3}$.

Moreover, the work of \cite{bh:sds,mbv:sds} confirms that the exponential decay assumption \eqref{eq:cor:decay:assumption} can be satisfied for smooth initial data. Their analysis builds up on the geometric scattering theory developed earlier in \cite{sabz:resonances,mm:mero}. Subsequently, Dyatlov \cite{sd:quasi,sd:beyond} established an exponential decay estimate in the stationary region of slowly rotating Kerr de Sitter black holes. His work extended \cite{wz:resolvent} and established an estimate holding up to and including the horizon, thus replacing the polynomial decay rate on the right hand side of \eqref{eq:cor:decay:assumption} by an exponential bound \eqref{eq:cor:decay:assumption:exp} with an arbitrary small loss of regularity at the level of the initial data. See also \cite{vd:microlocal,dv:resolvent} for a treatment of more general perturbations that applies to the stationary region of subextremal Kerr de Sitter spacetimes.

\end{rmk}

\subsubsection{General expanding cosmologies}
\label{sec:r:integral:gen}

The global redshift effect that underlies Proposition~\ref{prop:r:expanding} and the local redshift effect of Proposition~\ref{prop:sds:local} are present in a large class of expanding spacetimes which are ``close'' to the Schwarzschild de Sitter geometry. 
More precisely, the global redshift effect remains the stability mechanism in the cosmological region of spacetimes which
\begin{enumerate}
\item preserve the basic global \emph{causal} picture of the Schwarzschild de Sitter solution; in particular spacetimes that possess a domain that is bounded in the future by a \emph{spacelike} hypersurface $\Sigma^+$ at infinity, and in the past by two \emph{Killing horizons} $\mathcal{C}^+\cup\bar{\mathcal{C}}^+$, and
\item settle down asymptotically to a Schwarzschild de Sitter solution; in particular the metric and the deformation tensor of a timelike vectorfield converge uniformly to their counterparts of a Schwarzschild de Sitter solution as $\Sigma^+$ is approached.
\end{enumerate}
This class contains the family of subextremal Kerr de Sitter cosmologies. Note however, that \emph{no symmetries are required}. 
A detailed characterization of the class of metrics considered here is given in Def.~\ref{def:appl:class} and an explicit discussion of Kerr de Sitter metrics in this context is given in Section~\ref{sec:general}; in particular see Prop.~\ref{prop:kerr:G}.

\begin{thm}
  \label{thm:boundedness:general}
  Let $(\M,g)$ be a spacetime satisfying the assumptions of Definition~\ref{def:appl:class} with respect to a given (subextremal) Schwarzschild de Sitter solution $(\M,g_m)$; in particular $g$ and the deformation tensor $\pi$ of a fixed timelike coordinate vectorfield converge to its counterparts $g_m$ and $\pi_m$, respectively, as the future boundary $\Sigma^+$ is approached.
Then for any solution $\psi$ to the wave equation on $(\M,g)$ with finite (nondegenerate) energy on the cosmological horizons,
\begin{equation}
  D[\psi]=\int_{\mathcal{C}^+\cup\bar{\mathcal{C}}^+}{}^\ast J^N[\psi]<\infty\,,
\end{equation}
the energy of $\psi$ is uniformly bounded in the future of $\mathcal{C}^+\cup\bar{\mathcal{C}}^+$ as stated in Prop.~\ref{prop:appl:globalredshift}, and in particular the rescaled asymptotic energy on $\Sigma^+$ is finite:
  \begin{equation}
    \label{eq:thm:boundedness:rescaled:general}
    \int_{\Sigma^+}\sqv{\nablac\psi}\dm{\cg}\leq C D[\psi]
  \end{equation}
  where $C$ is a constant that only depends on the class of perturbations considered.  
  In particular, \eqref{eq:thm:boundedness:rescaled:general} holds true for subextremal Kerr de Sitter spacetimes.
\end{thm}

In the generality of spacetimes outlined above it is \emph{not} possible to establish the analogous statement of Proposition~\ref{prop:r:static} and to bound the energy on the right hand side of \eqref{eq:thm:boundedness:rescaled:general} by the energy on a spacelike hypersurface $\Sigma$ in the past of $\mathcal{C}_0^+$; c.f.~figure~\ref{fig:not:redshift:props}. The reason is that general perturbations of the static region $\mathcal{S}$ in Schwarzschild de Sitter contain an \emph{ergoregion} which gives way to the phenomenon of \emph{superradiance}; see \cite{dr:c} for further discussion. At present such a boundedness result is only available \emph{under the imposition of additional symmetries}: In \cite{dr:kerr:bounded} the problem of superradiance is overcome in the context of \emph{stationary} and \emph{axisymmetric} black hole spacetimes, and Dafermos and Rodnianski prove in particular a uniform bound for the nondegenerate energy on the event horizon. Their argument can be adapted to stationary axisymmetric perturbations of Schwarzschild \emph{de Sitter} black holes as noted in the concluding remarks of \cite{dr:kerr:bounded}. In this paper however we shall focus on the behaviour of linear waves in the expanding region which is tractable without a frequency decomposition even in the absence of symmetries.

Let us mention finally that the localisation result Cor.~\ref{cor:decay} applies verbatim to the class of spacetimes discussed above. In other words, whenever energy decay results are available along the cosmological horizons then these rates are inherited along the future boundary. Specifically for slowly rotating Kerr de Sitter spacetimes exponential decay rates were obtained by Dyatlov in \cite{sd:beyond}; see also \cite{vd:microlocal}.

\subsection{Pointwise estimates} 
The foliation $(\Sigma_r)$ is not only a natural choice for the energy estimates in this paper but also particularly convenient from the point of view of pointwise estimates. 
\subsubsection{Schwarzschild de Sitter spacetimes}
Note that in the spherically symmetric setting the tangent space to $\Sigma_r\ (r>\rc)$, and $\Sigma^+$ is spanned by Killing vectorfields, namely $T$ and $\Oi{i}:i=1,\ldots,3$. Since the commutations with these vectorfields are controlled in $\mathrm{L}^2$ by Thm.~\ref{thm:sds} we can immediately pass to pointwise estimates using the classical Sobolev embedding. Indeed, we can alternatively write the result \eqref{eq:thm:sds:energy} using the coercivity equality on the sphere (see Appendix \ref{a:coercivity}) as
\begin{equation}
  \int_{\Sigma^+}\Bigl\{\sq{T\cdot\psi}+\sum_{i=1}^3\sq{\Oi{i}\psi}\Bigr\}\dm{\cg}\leq C\,D[\psi]\,.
\end{equation}

\begin{cor}\label{thm:sds:pointwise}
Let $\Sigma$ be a $\Sigma^+$-Cauchy hypersurface as in Thm.~\ref{thm:sds}, and let $\psi$ be a solution to \eqref{eq:sds:wave} which in addition to \eqref{eq:thm:sds:condition} satisfies on $\Sigma$ the higher order energy condition
\begin{multline}
  D_c[\psi]=D[\psi]+\sum_{i=1}^3 D[\Oi{i}\psi]+\sum_{i,j=1}^3 D[\Oi{i}\Oi{j}\psi]\\
  +D[T\psi]+\sum_{i=1}^3 D[\Oi{i}T\psi]+\sum_{i,j=1}^3 D[\Oi{i}\Oi{j}T\psi]<\infty\,.
\end{multline}
Then we have the pointwise bound
\begin{equation}
  \label{eq:thm:pointwise:psi}
  \sup_{\Sigma^+}\lvert\psi\rvert\leq\sup_\Sigma\lvert\psi\rvert+C(\Lambda,m,\Sigma)\sqrt{D_c[\psi]}
\end{equation}
on the future boundary $\Sigma^+$, and the pointwise estimates
\begin{equation}
  \label{eq:thm:sds:pointwise}
  \sup_{\Sigma^+}\biggl\lvert\sq{\dd{t}{\psi}}+\sqv{\nablabc\psi}\biggr\rvert
  \leq C(\Lambda,m,\Sigma)\:D_c[\psi]\,,
\end{equation}
where $C$ is a constant that only depends on $\Sigma$, $\Lambda$, and $m$.
\end{cor}

\begin{rmk}
If the higher order energies associated to the solutions $\Oi{i}\psi$, $\Oi{i}\Oi{j}\psi$, and $T\psi$, $\Oi{i}T\psi$, $\Oi{i}\Oi{j}T\psi$ are assumed to decay along the cosmological horizon, similarly to Cor.~\ref{cor:decay}, with either a polynomial rate according to \eqref{eq:cor:decay:assumption} or exponentially according to \eqref{eq:cor:decay:assumption:exp}, then so does the quantity
\begin{equation}
  \label{eq:rmk:pointwise}
  \mathrm{PE}[\psi](\tau)\doteq\sup_{\Sigma^+_\tau}\biggl\lvert\sq{\dd{t}{\psi}}+\sqv{\nablabc\psi}\biggr\rvert\,,
\end{equation}
along $\Sigma^+$, according to \eqref{eq:cor:decay:conclusion} or \eqref{eq:cor:decay:conclusion:exp} respectively.
Moreover, pointwise decay for $\psi$ to a constant follows from \eqref{eq:rmk:pointwise} by integration on $\Sigma^+_\tau$.
\end{rmk}

\subsubsection{General expanding cosmologies} 
In the more general setting the tangent space to $\Sigma_r$ is spanned by approximate Killing vectorfields, namely vectorfields $T$, $\Oi{i}:i=1,2,3$ that asymptotically generate isometries at $\Sigma^+$. Under suitable assumptions on the rate of decay of the corresponding deformation tensors, the error terms generated by the commutations with $T$, and $\Oi{i}$ can be controlled by the global redshift effect. For brevity we shall forego the details of the analysis that yields analogous pointwise decay estimates to Cor.~\ref{thm:sds:pointwise} for more general expanding cosmologies.

Our result states in particular that solutions to the linear wave equation near- Schwarzschild de Sitter cosmologies have a limit on the future boundary that can be viewed as a function on the standard cylinder $\mathbb{R}\times\mathbb{S}^2$, and moreover, that the rescaled induced derivatives have a limit or equivalently that the geometrically induced derivatives of these solutions decay to the bounded derivatives of the limiting function. We have seen that as a consequence of the global redshift effect quantitative decay that is present along the horizon translates into a quantitative decay along the future boundary.

\section[Schwarzschild de Sitter spacetime]{Global geometry of the Schwarzschild de Sitter spacetime}

\input{geometry-sds}

\section[Local and global redshift]{Local and global redshift: \\Proofs of the boundedness and decay statements}
\label{sec:proofs:global}

We develop in this section the global energy estimates for solutions to the wave equation on Schwarzschild de Sitter spacetimes.

The crucial estimate for our analysis is proven in Section \ref{sec:sds:globalredshift} where we present a vectorfield that captures the \emph{global} redshift effect of the expanding region. In Section \ref{sec:sds:localredshift} the \emph{local} redshift effect of the cosmological horizons is exploited to control the energy in its immediate vicinity. We obtain a global estimate with the help of a boundedness result for the static region from \cite{dr:sds} as discussed above.

\subsection{Global redshift effect}
\label{sec:sds:globalredshift}

It is the purpose of this section to show that the vectorfield
\begin{equation}
  \label{eq:sds:M}
  M=\frac{1}{r}\pd{r}  
\end{equation}
captures the global redshift effect in the expanding region of Schwarzschild de Sitter.

The divergence theorem applied to the energy current \eqref{eq:r:ec} associated to the multiplier $M$, $J^M$, on
\begin{equation}
  \mathcal{R}_{r_1}^{r_2}\doteq\bigcup_{r_1\leq r\leq r_2}\Sigma_r\,,
\end{equation}
combined with the coarea formula applied to the foliation \eqref{eq:sds:metric:levelsets} reads
\begin{equation}
  \label{eq:sds:global:coarea}
    \int_{\Sigma_{r_2}}J^M\cdot n\,\dm{\gb{r_2}}+\int_{r_1}^{r_2}\ud r\int_{\Sigma_r}\phi \:\nabla\cdot J^M\dm{\gb{r}}=\int_{\Sigma_{r_1}}J^M\cdot n\,\dm{\gb{r_1}}\,,
\end{equation}
for any $r_2>r_1>\rc$.

The redshift property of $M$ is now manifested in the fact that for any solution $\psi$ to the wave equation \eqref{eq:sds:wave} it holds
  \begin{equation}\label{eq:M:redshift}
    \phi\:\nabla\cdot J^M[\psi]\geq\,\frac{1}{r}\:J^M[\psi]\cdot n\,,
  \end{equation}
which is the subject of the proposition below.
We are thus in the situation that
\begin{equation}\label{eq:sds:global:gronwall}
  \int_{\Sigma_{r_2}}J^M\cdot n\,\dm{\gb{r_2}}+\int_{r_1}^{r_2}\ud r\frac{1}{r}\int_{\Sigma_r}J^M\cdot n\,\dm{\gb{r}}\leq\int_{\Sigma_{r_1}}J^M\cdot n\,\dm{\gb{r_1}}\,,
\end{equation}
which implies by a Gronwall inequality:
\begin{equation}
  \label{eq:sds:global:weighted}
  r_2\int_{\Sigma_{r_2}}J^M\cdot n\,\dm{\gb{r_2}}\leq r_1\int_{\Sigma_{r_1}}J^M\cdot n\,\dm{\gb{r_1}}\qquad(r_2>r_1>\rc)\,.
\end{equation}
This is the content of Proposition \ref{prop:r:expanding}.

We begin with proving \eqref{eq:M:redshift}.
Recall here the conservation laws for solutions to the wave equation,
\begin{equation}
  \label{eq:T:conservation}
  \nabla\cdot T=0\,,
\end{equation}
namely
\begin{equation}
  \label{eq:energymomentum}
  T_{\mu\nu}[\psi]=\partial_\mu\psi\,\partial_\nu\psi-\frac{1}{2}\,g_{\mu\nu}\,\partial^\alpha\psi\partial_\alpha\psi
\end{equation}
is conserved provided $\psi$ is a solution to \eqref{eq:sds:wave}.
Therefore we obtain for the energy current associated to $M$,
\begin{equation}
  \label{eq:standardcurrent}
  J^M_\mu[\psi]=T_{\mu\nu}[\psi]M^\nu\,,
\end{equation}
that its divergence is given by
\begin{equation}
  \label{eq:KMd}
  \nabla\cdot J^M\doteq\nabla^\mu J_\mu^M={}^{(M)}\pi^{\mu\nu}T_{\mu\nu}[\psi]\doteq K^M[\psi]\,,
\end{equation}
where ${}^{(M)}\pi\doteq\frac{1}{2}\mathcal{L}_Mg$  is the deformation tensor of $M$.

\begin{prop}
  \label{prop:sds:globalredshift}
  Let $J^M$ be the current associated to the multiplier \eqref{eq:sds:M}, and $n$ the normal to $\Sigma_r$ $(r>\rc)$. Then for any solution $\psi$ to the wave equation \eqref{eq:sds:wave} we have
  \begin{equation*}
    \phi\:\nabla\cdot J^M[\psi]\geq\,\frac{1}{r}\:J^M[\psi]\cdot n
  \end{equation*}
on $\mathcal{R}^+$.
\end{prop}

\bigskip
\input{pf-globalredshift-sds}

It is now an immediate consequence of the following (slightly more general) Gronwall inequality that \eqref{eq:sds:global:gronwall} implies the weighted energy inequality \eqref{eq:sds:global:weighted}.

\begin{lemma}[Gronwall inequality for decreasing functions]
\label{lemma:gronwall:df}
Let $\alpha<0$, and $f,g\in\mathrm{C}^1([r_1,r_2])$ with $g\neq 0$, satisfying the inequality
\begin{equation}
  \label{eq:gronwall:assumption}
  f^\prime\leq\alpha\frac{g^\prime}{g}f
\end{equation}
on the interval $[r_1,r_2]$. Then
\begin{equation}
  \label{eq:lem:gronwall}
  f(r_2)\leq\frac{\lvert g(r_1)\rvert^{\lvert\alpha\rvert}}{\lvert g(r_2)\rvert^{\lvert\alpha\rvert}}f(r_1)\,.
\end{equation}
\end{lemma}

\begin{proof}
By \eqref{eq:gronwall:assumption},
\begin{equation}
  \frac{\ud}{\ud r}\biggl[f(r)\exp\Bigl[-\alpha\int_{r_1}^r\frac{g^\prime(r^\prime)}{g(r^\prime)}\ud r^\prime\Bigr]\biggr]\leq 0\,,
\end{equation}
which yields \eqref{eq:lem:gronwall} upon integration on the interval $[r_1,r_2]$.
\end{proof}

Finally, to verify that \eqref{eq:sds:global:weighted} implies the statement of Prop.~\ref{prop:r:expanding} note here
  \begin{multline}
    J^M\cdot n=\frac{1}{r}\frac{1}{\phi}T_{rr}
    =\frac{1}{2}\frac{1}{r^2}\phi\Bigl(\frac{\Lambda r^3}{3}+2m-r\Bigr)\sq{\dd{r}{\psi}}\\
    +\frac{1}{2}\phi\frac{1}{\frac{\Lambda r^3}{3}+2m-r}\sq{\dd{t}{\psi}}+\frac{1}{2}\frac{1}{r}\phi\sqv{\nablab\psi}\,,
  \end{multline}
which completes our proof of the energy estimate in the expanding region.

\subsection{Local redshift effect}
\label{sec:sds:localredshift}

In this section we construct a vectorfield that captures the local redshift effect of the cosmological horizon. The basic insight for this construction was gained in the context of black hole event horizons in \cite{dr:c}.

The idea is to construct a strictly timelike vectorfield $N$ in the vicinity of $\mathcal{C}^+$ which is Lie transported by $T$, $[T,N]=0$. As a consequence of the positivity of the surface gravity this can be done in such a way that the divergence of the current $J^N$ is positive. Thus, with $\Sigma$ and the notation as in Prop.~\ref{prop:sds:local}, the energy identity for $J^N$ on
\begin{equation*}
  \mathrm{J}^+(\Sigma^\prime\cup\mathcal{C}^+_0)\cap\mathrm{J}^-(\Sigma_{r_0})
\end{equation*}
implies for $r_0>\rc$ small enough (since the positive definiteness of the divergence only holds locally)
\begin{equation}
  \int_{\Sigma_{r_0}\cap\mathrm{J}^+(\Sigma)}J^N\cdot n\:\dm{\gb{r_0}}
  \leq\int_{\mathcal{C}^+_0}{}^\ast J^N + \int_{\Sigma^\prime}J^N\cdot n\:\dm{\gb{}}\,,
\end{equation}
which leads to Prop.~\ref{prop:sds:local} given the left hand side provides an estimate for the boundary terms arising from the current $J^M$ of the previous section. 

\input{pf-localredshift-sds}

\subsection{Limiting quantity}
\label{sec:limit}

While the \emph{bound} \eqref{eq:thm:sds:energy} follows from the global and local redshift estimates as discussed in Section~\ref{sec:r:integral}, in particular \eqref{eq:sds:overview:finiteintegral}, we remark that the \emph{existence} of the limit follows from a density argument. Any data $\psi\rvert_\Sigma\in\mathrm{H}^1(\Sigma)$ can be approximated by smooth functions $\underline{\psi}_j\in\mathrm{C}^\infty(\Sigma)$. Using the Sobolev inequality on the sphere at infinity (Lemma~\ref{lemma:sobolevembeddingsphere}), and commutations with the vectorfields $\Oi{i}$ \eqref{eq:def:Omegai}, we can show that all solutions $\psi_j$ to \eqref{eq:sds:wave} with initial data $\underline{\psi}_j$ on $\Sigma$ have a limit on $\Sigma^+$, or on $uv=1$ in the topology of $\mathbb{R}^2$ of the Penrose diagram in figure~\ref{fig:dn}, by integrating from $\Sigma$ to $\Sigma^+$ as discussed in Section~\ref{sec:sds:pointwise}. The trace of $\psi$ on $\Sigma^+$ is then approximated by the limit of $\psi_j$ in the norm defined by \eqref{eq:thm:sds:energy}, namely the homogeneous $\mathrm{H}^1$ norm on the cylinder $\mathbb{R}\times\mathbb{S}^2$. This completes the proof of Theorem~\ref{thm:sds}.

\subsection{Localisation}
\label{sec:sds:localisation}

In this section we show that quantitative energy decay along the cosmological horizon is propagated through the redshift region and inherited on the future boundary.

The observation that quantitative energy decay is preserved in a ``redshift region'' was first made by Dafermos in the context of the nonlinear Einstein-Maxwell scalar field model in spherical symmetry \cite{d:interior}. It was applied by Luk to the study of the wave equation in the language of vectorfield multipliers \cite{jl:improved}, who showed that any polynomial decay rate in the exterior of the Schwarzschild black hole is inherited on the event horizon (and in a neighborhood of the horizon in the interior of the black hole); it also appears in Dyatlov's argument \cite{sd:beyond} on exponential decay for solutions to the wave equation on Kerr de Sitter.

 The global redshift effect that is present in expanding spacetimes extends this mechanism beyond the vicinity of the cosmological horizon.
We give here a proof which avoids a recourse to a bootstrap argument, previously used in \cite{jl:improved}.

\input{localise}

\subsection{Pointwise estimates}
\label{sec:sds:pointwise}

The usefulness of the foliation $(\Sigma_r)$ is revealed in the Sobolev inequality for $\Sigma_r$ which allows us to pass immediately from energy estimates to pointwise estimates using the symmetries of the spacetimes $(\M,g)$.

\begin{prop}[Sobolev inequality on $\Sigma_r$]\label{prop:sds:sobolev}
Let $\psi\in\mathrm{H}^3(\Sigma_r)$, $r>\rc$, then
\begin{multline*}
  \sqrt{\frac{\Lambda r^2}{3}+\frac{2m}{r}-1}\:\:r^2\:\sup_{\Sigma} \psi^2\leq\\
  \leq C \int_{\Sigma}\biggl\{\psi^2+\sum_{i=1}^3\bigl(\Oi{i}\psi\bigr)^2+\sum_{i,j=1}^3\bigl(\Oi{i}\Oi{j}\psi\bigr)^2\\
  +\bigl(T\psi\bigr)^2+\sum_{i=1}^3\bigl(\Oi{i}T\psi\bigr)^2+\sum_{i,j=1}^3\bigl(\Oi{i}\Oi{j}T\psi\bigr)^2\biggr\}\dm{\gb{r}}\,.
\end{multline*}
on any submanifold $\Sigma\subset\Sigma_r$ of the form $I\times\mathrm{SO}(3)$, where $I\subset\mathbb{R}$ is an open interval with $\lvert I \rvert\geq 1$.
\end{prop}

The angular derivatives on the sphere are here estimated using the generators of the spherical isometries $\Oi{i}:i=1,2,3$. This fact is known as the coercivity inequality on the sphere which is discussed in Appendix \ref{a:coercivity}, where also the precise definition of the vectorfields $\Oi{i}$ can be found.

\begin{proof}
Recall the induced metric \eqref{eq:sds:inducedmetric} on $\Sigma_r$.
Let $t\,,\ t_0\in I$, and $\xi\in\mathbb{S}^2$, then
\begin{equation*}
  \lvert\psi^2(t;\xi)-\psi^2(t_0;\xi)\rvert\leq 2\int_{t_0}^t\lvert\psi\rvert\lvert\dd{t}{\psi}\rvert\ud t\,.
\end{equation*}
Therefore by Corollary \ref{cor:sobolev:spherer}, there exists a constant $C>0$ such that 
\begin{multline*}
  \lvert r^2 \psi^2(t)-r^2 \psi^2(t_0)\rvert\leq\\
  \leq C\int_I\int_{S_r}\lvert\psi\rvert^2+\sum_{i=1}^3\Bigl(\Oi{i}\psi\Bigr)^2+\sum_{i,j=1}^3\Bigl(\Oi{i}\Oi{j}\psi\Bigr)^2\dm{\gamma_r}\ud t\\
  +C\int_I\int_{S_r}\lvert T\psi\rvert^2+\sum_{i=1}^3\Bigl(\Oi{i}T\psi\Bigr)^2+\sum_{i,j=1}^3\Bigl(\Oi{i}\Oi{j}T\psi\Bigr)^2\dm{\gamma_r}\ud t\,.
\end{multline*}
Now choose $(t_0,\xi)\in I\times\mathbb{S}^2$ such that
\begin{equation*}
  r^2\psi^2(t_0;\xi)\leq\frac{1}{4\pi\lvert I\rvert}\int_I\int_{S_r}\psi^2\:\dm{\gamma_r}\ud t\,,
\end{equation*}
and we obtain the stated inequality in view of the expression \eqref{eq:sds:overview:volumeform} for the volume form of $\gb{r}$.
\end{proof}

The application of Prop.~\ref{prop:sds:sobolev} to the functions  $\phi^\frac{3}{2}T\cdot\psi$ and $\phi^\frac{1}{2}\frac{1}{r}\Oi{i}\psi:i=1,2,3$ yields quantities on the right hand side which are monotone by our results of Section~\ref{sec:sds:globalredshift}, provided $\psi$ is a solution to the wave equation \eqref{eq:sds:wave}. Indeed, we have
\begin{multline}
  \sup_{\Sigma_r}\biggl\{(r\phi)^2\,(T\psi)^2+r^2\lvert\nablab\psi\rvert_{\gamma_r}^2\biggr\}\leq\displaybreak[1]\\
  \leq C \int_{\Sigma_r}\biggl\{\phi^2\sq{T\psi}+\sqv{\nablab\psi}+\sum_{i=1}^3\phi^2\bigl(T\Oi{i}\psi\bigr)^2+\sum_{i=1}^3\sqv{\nablab\Oi{i}\psi}\\+\sum_{i,j=1}^3\phi^2\bigl(T\Oi{i}\Oi{j}\psi\bigr)^2+\sum_{i,j=1}^3\sqv{\nablab\Oi{i}\Oi{j}\psi}\biggr\}\,\phi\:\dm{\gb{r}}\displaybreak[0]\\
+C\int_{\Sigma_r}\biggl\{\phi^2\bigl(T^2\psi\bigr)^2+\sqv{\nablab T\psi}+\sum_{i=1}^3\phi^2\bigl(T^2\Oi{i}\psi\bigr)^2+\sum_{i=1}^3\sqv{\nablab T\Oi{i}\psi}\\+\sum_{i,j=1}^3\phi^2\bigl(T^2\Oi{i}\Oi{j}\psi\bigr)^2
  +\sum_{i,j=1}^3\sqv{\nablab T\Oi{i}\Oi{j}\psi}\biggr\}\,\phi\:\dm{\gb{r}}\,,
\end{multline}
which is bounded by Prop.~\ref{prop:r:expanding}, because $T$ and $\Oi{i}$ are Killing vectorfields, and thus $\Oi{i}\psi$, $\Oi{i}\Oi{j}\psi$, and $T\psi$, $T\Oi{i}\psi$, $T\Oi{i}\Oi{j}\psi$ are solutions to the wave equation. 

In view of
\begin{equation*}
  \lim_{r\to\infty}\: \sq{r\phi}=\frac{3}{\Lambda}\,,\qquad
  r^2\sqv{\nablab\psi}=\sqv{\nablabc\psi}\,,
\end{equation*}
the pointwise estimate \eqref{eq:thm:sds:pointwise} is now obtained analogously to Thm.~\ref{thm:sds} as discussed in Section~\ref{sec:r:integral}.

We may also apply Prop.~\ref{prop:sds:sobolev} to $\frac{1}{\phi}\dd{r}{\psi}\phi^\frac{1}{2}$ which yields by Prop.~\ref{prop:r:expanding} a monotone quantity. In fact we obtain,
\begin{equation}
  (r\phi)^{-2}\sup_{\Sigma_r}\sq{\dd{r}{\psi}}\leq\frac{C}{r^4}D_c[\psi]\,,
\end{equation}
which implies by integration from $\Sigma$ along $t$-constant lines the final bound \eqref{eq:thm:pointwise:psi}.

This completes the proof of Corollary~\ref{thm:sds:pointwise}.

\section{Applications to a general class of expanding spacetimes}
\label{sec:general}

\input{appl-gen}

\section{Applications to Klein-Gordon equations}
\label{sec:kg}

\input{appl-kg}

\section{Epilogue: The wave equation on de Sitter spacetimes}
\label{sec:ds}

The de Sitter cosmology is \emph{homogeneous} and can be cast as a hyperboloid in~$\mathbb{R}^5$ with spatially \emph{compact} sections $\mathbb{S}^3$.
The global study of solutions to the wave equation reduces in this setting in principle to a \emph{local} problem due to the presence of additional symmetries. 
In fact, it suffices to localise to the past of a point on the future boundary to obtain a global result; see also the work of Vasy~\cite{av:ds} and \cite{b:strichartz:a,baskin:parametrix,b:strichartz}.

In view of applications to Kerr de Sitter cosmologies however, we may entertain an approach to the global study of linear waves on de Sitter which \emph{does not make use of the homogeneity of the spacetime}. This consists in viewing the spacetime as spherically symmetric with respect to a \emph{fixed} timelike geodesic, and corresponds precisely to the case $m=0$ (see Section~\ref{sec:geometry:spherical:lambda}, or \cite{gh:ceh}). It does have the advantage that a global causal geometry emerges that shares some essential features with the Schwarzschild de Sitter geometry, as depicted in figure~\ref{figure:thm:ds}; in particular, there exists an expanding region in the sense described above, where the metric takes the form
\begin{equation}
  \label{eq:ds:intro:metric}
  g=-\phi^2\ud r^2+\gb{r}=-\frac{1}{\frac{\Lambda}{3}r^2-1}\ud r^2+\gb{r}\,,
\end{equation}
which is bounded in the past by cosmological horizons, and to the future by a spacelike hypersurface $\Sigma^+$ at infinity.

\begin{figure}[bt]
\includegraphics{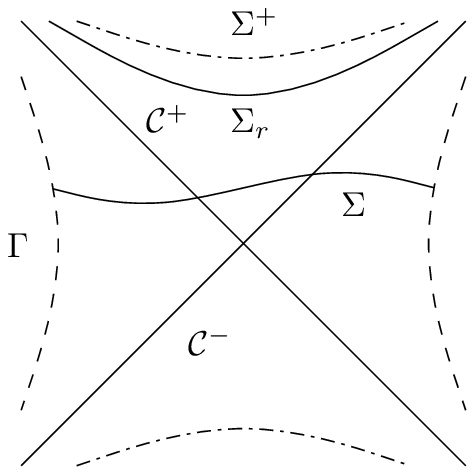}
\caption{Penrose diagram of de Sitter depicted as a spherically symmetric spacetime with respect to a chosen timelike geodesic $\Gamma$.}
\label{figure:thm:ds}
\end{figure}

This approach is carried out in all detail in my thesis \cite{vs:thesis}. The analysis yields results analogous to those of Section~\ref{sec:r}, and can form the basis of perturbative treatments. In particular we establish Theorem~\ref{thm:sds} in the case $m=0$.

Moreover, in \cite{vs:thesis} we provide an alternative proof of Prop.~\ref{prop:r:expanding} in double null coordinates, using the global redshift multiplier
\begin{equation}
  M=\frac{1}{1+\sqrt{\frac{\Lambda}{3}}r}\,\bigl(\bar{Y}+Y\bigr)\,,
  \qquad\text{where}\quad\bar{Y}=\frac{1}{\dd{u}{r}}\pd{u}\,,\quad Y=\frac{1}{\dd{v}{r}}\pd{v}\,,
\end{equation}
which highlights the global redshift effect as a suitable extension of the local redshift effect, (c.f.~Section~\ref{sec:sds:localredshift}).

\section{Acknowledgements}

This paper has grown out of work that I completed in significant parts in my thesis.
I would like to thank foremost my advisor Mihalis Dafermos for many insightful discussions and for suggesting this problem to me. I am very grateful for his support during my Ph.D. and beyond.
I have also benefited greatly from conversations with Willie Wong during his time at Cambridge, and more recently at MSRI.
The author thanks the Engineering Physical Sciences Research Council (UK), the Cambridge European Trust and the European Research Council for their financial support during his graduate studies, and the National Science Foundation (US) for their support of postdoctoral research at MSRI.

\appendix

\section{Coercivity inequality on the sphere}
\label{a:coercivity}

\input{acoercivity}

\section{Kerr de Sitter in comoving coordinates}
\label{a:kerr}

\input{kerrdesitter}


\input{glw.bbl}

\end{document}

%% file: geometry-sds.tex
The Schwarzschild de Sitter spacetimes $(\M,g)$ are a family of $3+1$-dimensional Lorentzian manifolds discovered independently by \cite{kottler:sds} and \cite{weyl:sds}, which in the range $\Lambda>0$, $0<3m<1/\sqrt{\Lambda}$ constitute our simplest model for a black hole in an expanding universe; ($m=0$ is de Sitter spacetime, our simplest cosmological model of an expanding universe \cite{nb:discovering}). They are unique in the sense of Birkhoff \cite{h:uniqueness,dr:c} as spherically symmetric solutions to the vacuum Einstein equations
\begin{equation}
  \label{eq:Einstein:c}
  R_{\mu\nu}-\frac{1}{2}g_{\mu\nu}R+\Lambda g_{\mu\nu}=0
\end{equation}
with positive cosmological constant $\Lambda>0$. 
While their global causal geometry is well known since \cite{gh:ceh,kr:effects}, a few remarks are in order concerning the coordinates used in this paper.

\subsection{Cosmological spacetimes under spherical symmetry.} 
\label{sec:geometry:spherical:lambda}
In spherical symmetry the topology of the spacetime manifold is necessarily that of $\mathcal{Q}\times\mathrm{SO}(3)$ where $\mathcal{Q}$ is a $1+1$-dimensional Lorentzian manifold. In a double null foliation of $\mathcal{Q}$ by the level sets of functions $u,v:\mathcal{Q}\to\mathbb{R}$ the metric $g$ assumes the form 
\begin{equation}
  \label{eq:metricsph}
  g=-\Omega^2\:\ud u\,\ud v+r^2\,\g\,,
\end{equation}
where $\g$ is the standard metric of the unit sphere $\mathbb{S}^2$, and we can think of \eqref{eq:Einstein:c} as partial differential equations for the area radius $r$ and the conformal factor $\Omega$ as functions of $u$, and $v$. Indeed, it is easily deduced from \eqref{eq:Einstein:c} that the area radius $r$ satisfies the Hessian equations\begin{subequations}
\label{eq:hessiannull}
\begin{gather}
\frac{\partial^2 r}{\partial u^2}-\frac{2}{\Omega}\frac{\partial\Omega}{\partial u}\frac{\partial r}{\partial u}=0\\
\frac{\partial^2 r}{\partial u\,\partial v}+\frac{1}{r}\frac{\partial r}{\partial u}\frac{\partial r}{\partial v}=-\frac{\Omega^2}{4r}+\frac{\Omega^2}{4}\Lambda r\\
\frac{\partial^2 r}{\partial v^2}-\frac{2}{\Omega}\frac{\partial\Omega}{\partial v}\frac{\partial r}{\partial v}=0\,.
\end{gather}
\end{subequations}
We observe that the mass function $m:\mathcal{Q}\to\mathbb{R}$ implicitly defined by (and motivated in \cite{ch:fluids:sym})
\begin{equation}
  \label{eq:massfunction}
  1-\frac{2m}{r}-\frac{\Lambda r^2}{3}=-\frac{4}{\Omega^2}\dd{u}{r}\dd{v}{r}\,,
\end{equation}
is constant by virtue of \eqref{eq:hessiannull}, and precisely the quantity that parametrizes the Schwarzschild de Sitter family for any fixed $\Lambda>0$.
It is useful to introduce the ``tortoise coordinate''
\begin{equation}
  \label{eq:sds:tortoise}
  r^\ast=\int\frac{1}{1-\frac{2m}{r}-\frac{\Lambda r^2}{3}}\ud r\,,
\end{equation}
which satisfies again as a consequence of \eqref{eq:hessiannull} the simple partial differential equation
\begin{equation}
  \label{eq:pde:rs}
  \frac{\partial^2 r^\ast}{\partial u\,\partial v}=0\,.
\end{equation}
In general, the dependence of the area radius $r$ on the null coordinates $u$, $v$, is thus of the form
\begin{equation}
  \label{eq:sds:rsol:gen}
  r^\ast=f_1(u)+f_2(v)\,,
\end{equation}
where we are free to choose the real valued functions $f_1$, and $f_2$.
We may think of the different charts on $(\M,g)$ to be obtained with suitable choices of the functions $f_1$, $f_2$, and of the constants for the indefinite integral \eqref{eq:sds:tortoise}.
While in the case $m=0$ the manifold can be covered with a single chart, this is not possible for $m>0$. Here we discuss a chart that covers the cosmological regions and the adjacent static domains.

\subsection{Double null coordinates for the cosmological and static regions}
\label{sec:sds:dn}

The horizons of Schwarzschild de Sitter are null hypersurfaces consisting of spheres of constant area radius, and thus null lines on $\mathcal{Q}$ on which the right hand side of \eqref{eq:massfunction} vanishes. The polynomial on the left hand side of \eqref{eq:massfunction} has three distinct real roots $\rcb$, $\rH$, and $\rc$ provided
\begin{equation}
  \label{eq:sds:parconstraint}
  0<3m<\frac{1}{\sqrt{\Lambda}}\,,
\end{equation}
satisfying (explicit expressions may be found in \cite{kr:effects}):
\begin{equation}
  \label{eq:sds:roots:location}
  \rcb<0<2m<\rH<3m<\rc\,.
\end{equation}
We note
\begin{equation}
  \label{eq:sds:polyfactors}
  r-2m-\frac{\Lambda r^3}{3}=-\frac{\Lambda}{3}(r-\rH)(r-\rc)(r+\lvert\rcb\rvert)\,,
\end{equation}
with
\begin{subequations}
\label{eq:sds:roots:rel}
\begin{gather}
  -\vrcb+\rc+\rH=0\\
  \rc\vrcb+\rH\vrcb-\rH\rc=\frac{3}{\Lambda}\\
  \rH\rc\vrcb=\frac{6m}{\Lambda}\,,
\end{gather}
\end{subequations}
and by decomposition into partial fractions
\begin{multline}
  \label{eq:sds:partialfractions}
  \frac{1}{1-\frac{2m}{r}-\frac{\Lambda r^2}{3}}=\frac{3}{\Lambda}\frac{\rH}{(\rc-\rH)(\rH+\vrcb)}\frac{1}{r-\rH}\\
  -\frac{3}{\Lambda}\frac{\rc}{(\rc-\rH)(\rc+\vrcb)}\frac{1}{r-\rc}\\+\frac{3}{\Lambda}\frac{\vrcb}{(\vrcb+\rH)(\vrcb+\rc)}\frac{1}{r+\vrcb}\,.
\end{multline}
The chart covering the cosmological horizons $r=\rc$ is now obtained by choosing \eqref{eq:sds:tortoise} to be centred at $r=3m$, 
\begin{equation}
  \label{eq:tortoise:centred}
  r^\ast(r)=\int_{3m}^r\frac{1}{1-\frac{2m}{r}-\frac{\Lambda r^2}{3}}\ud r\,,
\end{equation}
and setting 
\begin{subequations}
\label{eq:sds:rsol}
\begin{equation}
  f_1(x)=f_2(x)=-\frac{3}{\Lambda}\frac{\rc}{(\rc-\rH)(\rc+\vrcb)}\log\frac{\lvert x\rvert}{A}
\end{equation}
\end{subequations}
where $A$ is the constant
\begin{equation}
  A^2=(\rc-3m)(3m-\rH)^{-\frac{\rH}{\rc}\frac{\rc+\vrcb}{\rH+\vrcb}}(3m+\vrcb)^{-\frac{\vrcb}{\rc}\frac{\rc-\rH}{\vrcb+\rH}}\,.
\end{equation}
Indeed, the integration of \eqref{eq:tortoise:centred} using \eqref{eq:sds:partialfractions} yields in view of \eqref{eq:sds:rsol} the following relation between the chosen null coordinates $(u,v)$ and the radius function $r$:
\begin{equation}
  \label{eq:sds:uvr}
  uv=\frac{r-\rc}{(r-\rH)^{\frac{\rH}{\rc}\frac{\rc+\vrcb}{\rH+\vrcb}}(r+\vrcb)^{\frac{\vrcb}{\rc}\frac{\rc-\rH}{\vrcb+\rH}}}\,.
\end{equation}
We recover precisely the picture of figure \ref{fig:dn} where the cosmological horizons $r=\rc$ are located at $u=0$ and $v=0$, while the surfaces of constant $r\neq\rc$ are spacelike hyperbolas in the $uv$-plane for $r>\rc$ and timelike hyperbolas for $r<\rc$. Moreover the future timelike boundary $r=\infty$ is identified with the spacelike hyperbola $uv=1$.

Since
\begin{subequations}
  \begin{gather}
    \dd{u}{r^\ast}=\frac{1}{1-\frac{2m}{r}-\frac{\Lambda r^2}{3}}\dd{u}{r}=-\frac{3}{\Lambda}\frac{\rc}{(\rc-\rH)(\rc+\vrcb)}\frac{1}{u}\\
    \dd{v}{r^\ast}=\frac{1}{1-\frac{2m}{r}-\frac{\Lambda r^2}{3}}\dd{v}{r}=-\frac{3}{\Lambda}\frac{\rc}{(\rc-\rH)(\rc+\vrcb)}\frac{1}{v}    
  \end{gather}
\end{subequations}
we can solve \eqref{eq:massfunction} for $\Omega^2$ to obtain:
\begin{multline}
  \label{eq:sds:geometry:omega}
  \Omega^2=\frac{4}{r}\frac{3}{\Lambda}\frac{\rc^2}{(\rc-\rH)^2(\rc+\vrcb)^2}\times\\
  \times\bigl(r-\rH\bigr)^{1+\frac{\rH}{\rc}\frac{\rc+\vrcb}{\rH+\vrcb}}\bigl(r+\vrcb\bigr)^{1+\frac{\vrcb}{\rc}\frac{\rc-\rH}{\vrcb+\rH}}
\end{multline}
The (nondegenerate) metric $g$ on the chart that covers the region $\rH<r<\infty$ and extends across the cosmological horizon $r=\rc$ thus takes in double null coordinates $(u,v)$ the form
\begin{multline}
  \label{eq:sds:metric:doublenull}
  g=-\frac{4}{r}\frac{3}{\Lambda}\frac{\rc^2}{(\rc-\rH)^2(\rc+\vrcb)^2}\times\\
  \times\bigl(r-\rH\bigr)^{1+\frac{\rH}{\rc}\frac{\rc+\vrcb}{\rH+\vrcb}}\bigl(r+\vrcb\bigr)^{1+\frac{\vrcb}{\rc}\frac{\rc-\rH}{\vrcb+\rH}}\ud u\ud v+r^2\g
\end{multline}
where $r$ is a function of $(u,v)$ implicity given by \eqref{eq:sds:uvr}.

\subsection{Constant area radius foliation of the expanding region}
\label{sec:sds:foliation}

We have seen that the expanding region
\begin{equation}
  \label{eq:sds:geometry:expanding}
  \mathcal{R}^+=\Bigl\{(u,v):0<uv<1\Bigr\}
\end{equation}
can be foliated in a geometrically natural way by spacelike hypersurfaces $\Sigma_r$ of constant area radius $r$,
\begin{equation}
  \label{eq:sds:geometry:foliate}
  \mathcal{R}^+=\bigcup_{\rc<r<\infty}^\infty\Sigma_r\,.
\end{equation}
Here the area radius serves as a time function on $\mathcal{R}^+\subset\M$ ($r$ is strictly increasing along any future-directed timelike curve in $\mathcal{R}^+$).
We obtain the decomposition of the metric $g$ with respect to the foliation \eqref{eq:sds:geometry:foliate} with the help of the gradient vectorfield
\begin{equation}
  \label{eq:sds:r:grad}
  V^\mu=-g^{\mu\nu}\partial_\nu r\,.
\end{equation}
Indeed, since 
\begin{equation}
  \label{eq:sds:V}
  V=\frac{1}{2}\frac{\Lambda}{3}\frac{(\rc-\rH)(\rc+\vrcb)}{\rc}\Bigl(u\pd{u}+v\pd{v}\Bigr)\,,
\end{equation}
we find for the lapse function
\begin{equation}
  \label{eq:sds:lapse}
  \phi\doteq\frac{1}{\sqrt{-g(V,V)}}=\frac{1}{\sqrt{\frac{\Lambda r^2}{3}+\frac{2m}{r}-1}}
\end{equation}
and the metric takes the form \eqref{eq:sds:metric:levelsets}, namely
\begin{equation}
  g=-\phi^2\ud r^2+\gb{r}\,.
\end{equation}
Moreover, we have \eqref{eq:sds:dr}.
The explicit form of the induced metric $\gb{r}$ on $\Sigma_r$ is readily inferred from \eqref{eq:sds:metric:doublenull} upon the introduction of the coordinate\begin{equation}
  \label{eq:sds:t}
  t=-\frac{3}{\Lambda}\frac{\rc}{(\rc+\vrcb)(\rc-\rH)}\log\lvert\frac{v}{u}\rvert\,.
\end{equation}
We can then read off \eqref{eq:sds:inducedmetric} from the explicit expression for the metric $g$ in $(t,r)$ coordinates, which reads
\begin{equation}
  \label{eq:sds:metric:tr}
  g=-\frac{1}{\frac{\Lambda r^2}{3}+\frac{2m}{r}-1}\ud r^2+\Bigl(\frac{\Lambda r^2}{3}+\frac{2m}{r}-1\Bigr)\ud t^2+r^2\g\,.
\end{equation}
Note here that it follows from \eqref{eq:sds:uvr} and \eqref{eq:sds:t} that
\begin{equation}
  \label{eq:sds:dt}
  \pd{t}=\frac{1}{2}\frac{\Lambda}{3}\frac{(\rc+\vrcb)(\rc-\rH)}{\rc}\Bigl(u\pd{u}-v\pd{v}\Bigr)\,,
\end{equation}
and (which we have already made use of in \eqref{eq:sds:V})
\begin{subequations}
  \label{eq:sds:geometry:duvr}
  \begin{gather}
    \dd{u}{r}=\frac{1}{4}\frac{\Lambda}{3}\frac{(\rc+\vrcb)(\rc-\rH)}{\rc}\Omega^2\,v\\
    \dd{v}{r}=\frac{1}{4}\frac{\Lambda}{3}\frac{(\rc+\vrcb)(\rc-\rH)}{\rc}\Omega^2\,u\,.
  \end{gather}
\end{subequations}

\subsection{Surface gravity of the cosmological horizon}
\label{sec:geometry:horizon}
In due course the positivity of the surface gravity of the cosmological horizons will be of crucial importance. While this fact in itself is well-known \cite{gh:ceh}, we take its statement as an opportunity to introduce a vectorfield $Y$ which is the starting point for all local redshift multiplier constructions (see \cite{dr:c}).

The observation is that the vectorfield \eqref{eq:sds:dt} extends to a global vectorfield that characterizes the cosmological horizon as a Killing horizon with positive surface gravity.
\begin{lemma}[Positive surface gravity of the cosmological horizons]
  \label{lemma:sds:surfacegravity}
  The vectorfield
  \begin{equation}
    \label{eq:sds:T}
    T\doteq\kC\Bigl(u\pd{u}-v\pd{v}\Bigr)
  \end{equation}
  is globally Killing, i.e.~
  \begin{equation}
    \label{eq:sds:T:Killing}
    \deformt{T}\doteq\frac{1}{2}\mathcal{L}_Tg=0\,,
  \end{equation}
 and satisfies
 \begin{equation}
   \label{eq:sds:T:sg}
   \nabla_TT=\kC T\qquad\text{on }\mathcal{C}^+\,,
 \end{equation}
 where
 \begin{equation}
   \label{eq:sds:kC}
   \kC=\frac{1}{2}\frac{\Lambda}{3}\frac{(\rc-\rH)(\rc+\vrcb)}{\rc}>0
 \end{equation}
 is the surface gravity of the cosmological horizons.
\end{lemma}
The result is obtained with the help of the vectorfield
\begin{equation}
  \label{eq:sds:Y:C}
  Y\Bigr\rvert_{\mathcal{C}^+}=\frac{2}{\dd{v}{r}}\pd{v}
\end{equation}
which is conjugate to $T$ along the horizon (i.e.~$Y$ is null, orthogonal to the tangent space of the spheres of symmetry, and verifies the following normalization):
\begin{equation}
  \label{eq:sds:T:Y}
  g(T,Y)\VC=-2\,.
\end{equation}
Indeed, by \eqref{eq:sds:T:Killing},
\begin{multline}
  g(\nabla_TT,Y)\VC=-g(\nabla_YT,T)\VC=-\frac{1}{2}Y\cdot g(T,T)\VC\\
  =\frac{\ud}{\ud r}\Bigl(1-\frac{2m}{r}-\frac{\Lambda r^2}{3}\Bigr)\Bigr\rvert_{r=\rc}\,.
\end{multline}
Alternatively $\kC$ is characterized by
\begin{equation}
  \label{eq:sds:kC:alternate}
  \nabla_YT=\nabla_TY=-\kC Y\qquad\text{on }\mathcal{C}^+\,;
\end{equation}
note that this in particular implies that $Y$ is Lie transported by $T$ along the horizon:
\begin{equation}
  \label{eq:sds:LieTY}
  [T,Y]\VC=0\,.
\end{equation}
\begin{rmk}
The vectorfield $Y$ as defined above forms the basis of an extension away from the horizon which captures the redshift effect of horizons with positive surface gravity. It was first introduced in the context of the Schwarzschild black hole spacetime in \cite{dr:redshift,dr:c}.
\end{rmk}

%% file: pf-globalredshift-sds.tex
\noindent\emph{Proof.}
  It is equivalent to establish
  \begin{equation*}
    \label{eq:sds:ee:proof:div}
    K^M=\deformt{M}^{\mu\nu}T_{\mu\nu}[\psi]\geq\frac{1}{r^2}\frac{1}{\phi^2}T_{rr}[\psi]\,.
  \end{equation*}

  Recall from \eqref{eq:sds:metric:tr} that
  \begin{subequations}
    \begin{gather*}
      g_{rr}=-\frac{3}{\Lambda}\frac{r}{(r-\rH)(r-\rc)(r+\vrcb)}\\
      g_{tt}=\frac{1}{r}\frac{\Lambda}{3}(r-\rH)(r-\rc)(r+\vrcb)\\
      g_{AB}=r^2\g_{AB}\,,
    \end{gather*}
  \end{subequations}
  and thus the non-vanishing connection coefficients are:
  \begin{subequations}
    \begin{gather*}
      \begin{split}
      \Gamma_{rr}^r&=\frac{1}{2}(g^{-1})^{rr}\partial_rg_{rr}\\
      =&\phantom{+}\frac{1}{2}\frac{1}{r}\frac{\rH}{(\rc-\rH)(\rH+\vrcb)}\frac{(r-\rc)(r+\vrcb)}{r-\rH}\\
      &-\frac{1}{2}\frac{1}{r}\frac{\rc}{(\rc-\rH)(\rc+\vrcb)}\frac{(r-\rH)(r+\vrcb)}{r-\rc}\\
      &+\frac{1}{2}\frac{1}{r}\frac{\vrcb}{(\vrcb+\rH)(\vrcb+\rc)}\frac{(r-\rH)(r-\rc)}{r+\rH}
    \end{split}\displaybreak[0]\\
    \Gamma_{rt}^t=\frac{1}{2}(g^{-1})^{tt}\partial_r g_{tt}=\frac{1}{2}\Bigl[-\frac{1}{r}+\frac{1}{r-\rH}+\frac{1}{r-\rc}+\frac{1}{r+\vrcb}\Bigr]\,,
  \end{gather*}
  \end{subequations}
  and $\Gamma_{tt}^r$ as well as
  \begin{subequations}
    \begin{gather*}
      \Gamma_{AB}^r=\frac{\Lambda}{3}(r-\rH)(r-\rc)(r+\vrcb)\g_{AB}\\
      \Gamma_{rA}^B=\frac{1}{r}\delta_A^B\,.
    \end{gather*}
  \end{subequations}

  Let us first consider
  \begin{equation*}
    \label{eq:sds:Mp}
    \Mp=\pd{r}\,.
  \end{equation*}
  We have
  \begin{subequations}
    \begin{gather*}
      \deformt{\Mp}^{rr}=(g^{-1})^{rr}\Gamma_{rr}^r\qquad\deformt{\Mp}^{tt}=(g^{-1})^{tt}\Gamma_{rt}^t\\
      \deformt{\Mp}^{AB}=\frac{1}{r}(g^{-1})^{AB}\,,
    \end{gather*}
  \end{subequations}
  and
  \begin{equation*}
    \label{eq:sds:ee:proof:bulk}
    K^{\Mp}=\deformt{\Mp}^{rr}T_{rr}+\deformt{\Mp}^{tt}T_{tt}+\deformt{\Mp}^{AB}T_{AB}\,.
  \end{equation*}
  Now,
  \begin{multline*}
    T_{rr}=\frac{1}{2}\sq{\dd{r}{\psi}}+\frac{1}{2}\sq{\frac{3}{\Lambda}}\frac{r^2}{(r-\rH)^2(r-\rc)^2(r+\vrcb)^2}\sq{\dd{t}{\psi}}\\
      +\frac{1}{2}\frac{3}{\Lambda}\frac{r}{(r-\rH)(r-\rc)(r+\vrcb)}\sqv{\nablab\psi}\,,
  \end{multline*}
  \begin{multline*}
    T_{tt}=\frac{1}{2}\sq{\dd{t}{\psi}}+\frac{1}{2}\frac{1}{r^2}\sq{\frac{\Lambda}{3}}(r-\rH)^2(r-\rc)^2(r+\vrcb)^2\sq{\dd{r}{\psi}}\\
    -\frac{1}{2}\frac{1}{r}\frac{\Lambda}{3}(r-\rH)(r-\rc)(r+\vrcb)\sqv{\nablab\psi}\,,
  \end{multline*}
  and
  \begin{multline*}
    g^{AB}T_{AB}=\frac{\Lambda}{3}\frac{1}{r}(r-\rH)(r-\rc)(r+\vrcb)\sq{\dd{r}{\psi}}\\
    -\frac{3}{\Lambda}\frac{r}{(r-\rH)(r-\rc)(r+\vrcb)}\sq{\dd{t}{\psi}}\,;
  \end{multline*}
  also note that
  \begin{multline*}
    \frac{1}{\phi^2}T_{rr}=\frac{1}{2}\frac{1}{r}\frac{\Lambda}{3}(r-\rH)(r-\rc)(r+\vrcb)\sq{\dd{r}{\psi}}\\
    +\frac{1}{2}\frac{3}{\Lambda}\frac{r}{(r-\rH)(r-\rc)(r+\vrcb)}\sq{\dd{t}{\psi}}+\frac{1}{2}\sqv{\nablab\psi}\\
    \doteq \frac{1}{2}L_r\sq{\dd{r}{\psi}}+\frac{1}{2}L_t\sq{\dd{t}{\psi}}+\frac{1}{2}\sqv{\nablab\psi}\,.
  \end{multline*}
  We then find that
  \begin{multline*}
    K^\Mp=\frac{1}{2}\Bigl[K_0+K_1+K_2\Bigr]L_r\sq{\dd{r}{\psi}}\\
    +\frac{1}{2}\Bigl[K_0+K_1-K_2\Bigr]L_t\sq{\dd{t}{\psi}}
    +\frac{1}{2}\Bigl[K_0-K_1\Bigr]\sqv{\nablab\psi}\,,
  \end{multline*}
  where
  \begin{multline*}
    K_0=-\frac{1}{2}\frac{1}{r}\frac{\rH}{(\rc-\rH)(\rH+\vrcb)}\frac{(r-\rc)(r+\vrcb)}{r-\rH}\\
    +\frac{1}{2}\frac{1}{r}\frac{\rc}{(\rc-\rH)(\rc+\vrcb)}\frac{(r-\rH)(r+\vrcb)}{r-\rc}\\
    -\frac{1}{2}\frac{1}{r}\frac{\vrcb}{(\vrcb+\rH)(\vrcb+\rc)}\frac{(r-\rH)(r-\rc)}{r+\vrcb}\displaybreak[0]\\
    =-\frac{1}{2}\Bigl(\frac{\Lambda r^2}{3}+\frac{2m}{r}-1\Bigr)\pd{r}\frac{1}{\frac{\Lambda r^2}{3}+\frac{2m}{r}-1}\\
    =\frac{1}{2}\frac{3}{\Lambda}\frac{2\frac{\Lambda r^2}{3}-\frac{2m}{r}}{(r-\rH)(r-\rc)(r+\vrcb)}\,,
  \end{multline*}
  and
  \begin{equation*}
    K_1=\frac{1}{2}\frac{1}{r}\Bigl(-1+\frac{r}{r-\rH}+\frac{r}{r-\rc}+\frac{r}{r+\vrcb}\Bigr)
  \end{equation*}
and finally
\begin{equation*}
  K_2=2\frac{1}{r}\,.
\end{equation*}
Next we observe that
\begin{multline*}
  K_0-K_1=\frac{1}{2}\frac{1}{r}\frac{1}{(r-\rH)(r-\rc)(r+\vrcb)}\times\displaybreak[0]\\
  \times\Bigl[2r^3-\frac{6m}{\Lambda}+(r-\rH)(r-\rc)(r+\vrcb)\\
  -r\bigl(r^2+\vrcb r-\rc r-\rc\vrcb\bigr)\\
  -r\bigl(r^2+\vrcb r-\rH r-\rH\vrcb\bigr)\\-r\bigl(r^2-\rc r-\rH r+\rH\rc\bigr)\Bigr]=0
\end{multline*}
by \eqref{eq:sds:roots:rel}, and thus
\begin{multline*}
  K_0+K_1-K_2=2 K_1-K_2=\frac{1}{r}\Bigl(-3+\frac{r}{r-\rH}+\frac{r}{r-\rc}+\frac{r}{r+\vrcb}\Bigr)\displaybreak[0]\\
  =\frac{1}{r}\frac{1}{(r-\rH)(r-\rc)(r+\vrcb)}\Bigl[2\rH\vrcb r+2\rc\vrcb r-2\rH\rc r-3\rc\vrcb\rH\Bigr]\\
  =\frac{3}{\Lambda}\frac{2}{r}\frac{1}{(r-\rH)(r-\rc)(r+\vrcb)}\Bigl[r-3m\Bigr]\geq\frac{3}{\Lambda}\frac{2}{r}\frac{1}{(r-\rH)(r+\vrcb)}\,,
\end{multline*}
again using the relations \eqref{eq:sds:roots:rel}; and finally
\begin{equation*}
  K_0+K_1+K_2\geq 2 K_2=4\frac{1}{r}\,.
\end{equation*}
We have shown in particular
\begin{equation*}
  K^\Mp\geq 0\,.
\end{equation*}
Since
\begin{subequations}
  \begin{gather*}
    \deformt{M}_{rr}=-\frac{1}{r^2}\,g_{rr}+\frac{1}{r}\,\deformt{\Mp}_{rr}\\
    \deformt{M}_{tt}=\frac{1}{r}\,\deformt{\Mp}_{tt}\\
    \deformt{M}_{AB}=\frac{1}{r}\,\deformt{\Mp}_{AB}\,,
  \end{gather*}
\end{subequations}
we conclude
\begin{multline*}
  K^M=-\frac{1}{r^2}(g^{-1})^{rr}T_{rr}+\frac{1}{r}K^\Mp\\
  =\frac{\Lambda}{3}\frac{1}{r^3}(r-\rH)(r-\rc)(r+\vrcb)\,T_{rr}+\frac{1}{r}K^\Mp
  \geq\frac{1}{r^2}\frac{1}{\phi^2}T_{rr}\,.
\end{multline*}
\qed

%% file: pf-localredshift-sds.tex
We shall now define $N$ in the double null coordinates of Section \ref{sec:sds:dn}.
On $\mathcal{C}^+$ we may introduce the null frame $(T,Y,E_A:A=1,2)$, made up of the vectorfields $T$ \eqref{eq:sds:T}, $Y\onc$ as discussed in Section \ref{sec:geometry:horizon}, and complemented by an orthonormal frame field  $E_A:A=1,2$ on the sections of the horizon which is Lie transported by $T$.
We have seen that, also using \eqref{eq:sds:geometry:duvr},
\begin{subequations}\label{eq:sds:local:nullframe:horizon} 
\begin{gather}
  T\VC=\kC\,u\pd{u}\,,\\
  Y\VC=\frac{2}{\dd{v}{r}\rvert_{v=0}}\pd{v}=\ic\frac{1}{u}\pd{v}\,,\qquad\ic\doteq\frac{4}{\kC}\frac{1}{\Omega^2}\VC\,,
\end{gather}
\end{subequations}
are conjugate null vectors on the cosmological horizon
\begin{subequations}
\begin{gather}
  g(T,Y)\VC=\ic\kC\,g_{uv}=-2\,,\\
  g(T,E_A)\VC=0\qquad g(Y,E_A)\VC=0\,,
\end{gather}
\end{subequations}
and satisfy the commutation relation \eqref{eq:sds:LieTY}.

Let us now \emph{extend} the vectorfield $Y$ away from the horizon by
\begin{equation}
  \label{eq:sds:redshift:extension}
  \nabla_YY=-\sigma(Y+T)\,,
\end{equation}
where $\sigma>0$.
It is well known from \cite{dr:c} that this extension gives rise to a positive current by virtue of Lemma~\ref{lemma:sds:surfacegravity}. Indeed, we find using \eqref{eq:sds:kC:alternate} and  \eqref{eq:sds:redshift:extension} that
\begin{multline}
  K^Y\VC\doteq\deformt{Y}^{\mu\nu}T_{\mu\nu}\VC=\frac{1}{2}\sigma\sq{T\cdot\psi}+\frac{1}{2}\kC\sq{Y\cdot\psi}\\
  +\frac{1}{2}\sigma\sqv{\nablab\psi}+\frac{2}{r}\bigl( T\cdot\psi\bigr)\bigl(Y\cdot \psi\bigr)\,,
\end{multline}
and thus
\begin{equation}
  \label{eq:sds:redshift:div}
  K^Y\VC\geq\frac{1}{4}\kC\sq{Y\cdot\psi}+\Bigl[\frac{1}{2}\sigma-\sq{\frac{2}{\rc}}\frac{1}{\kC}\Bigr]\sq{T\cdot\psi}+\frac{1}{2}\sigma\sqv{\nablab\psi}\geq 0\,,
\end{equation}
if $\sigma>(2/\rc)^22/\kC$. 
By construction $Y$ is Lie transported by $T$,
\begin{equation}\label{eq:sds:local:Y:invariance}
  [T,Y]=0\,.
\end{equation}

\begin{prop}
  \label{prop:sds:redshift}
Let $\Sigma$, $\Sigma^\prime$, and $\mathcal{C}_0^+$ be as in Prop.~\ref{prop:sds:local} and let $\psi$ be a solution to \eqref{eq:sds:wave}.
Set $N=T+Y$ with $Y$ as constructed above, then there exists $r_0>\rc$ (depending on $\Lambda$, $m$ and $\Sigma$) such that
  \begin{equation}
    \label{eq:sds:prop:redshift:positivity}
    K^N[\psi]\geq 0\qquad\text{ on }\:\mathrm{J}^+(\Sigma^\prime\cup\mathcal{C}^+_0)\cap\mathrm{J}^-(\Sigma_{r_0})\,,
  \end{equation}
  and a constant $C$ that only depends on $r_0$ such that
  \begin{equation}
    \label{eq:sds:prop:redshift:energy}
    J^M[\psi]\cdot n\leq C(r_0)\,J^N[\psi]\cdot n\qquad\text{ on }\:\Sigma_{r_0}\cap\mathrm{J}^+(\Sigma)\,.
  \end{equation}
\end{prop}

The positivity of the divergence \eqref{eq:sds:prop:redshift:positivity} follows by continuity from \eqref{eq:sds:redshift:div} with $r_0-\rc$ chosen small enough, and $\sigma$ as indicated above. Also by continuity $N$ remains timelike in a neighborhood of $\mathcal{C}^+$ for $r_0-\rc>0$ small enough, which implies \eqref{eq:sds:prop:redshift:energy}; note that the uniformity of the constant $C$ in $r$ follows from the invariance of $N$ under the flow generated by $T$, namely \eqref{eq:sds:local:Y:invariance}.

Alternatively, we can solve the O.D.E.~\eqref{eq:sds:redshift:extension} with initial conditions \eqref{eq:sds:local:nullframe:horizon} in the double null coordinates of Section \ref{sec:sds:dn}. If we set $\sigma=2\sq{\frac{2}{\rc}}\frac{2}{\kC}$, we obtain
\begin{multline}
  \label{eq:sds:local:N:explicit}
  N=\kC\Bigl[1-\frac{2}{\kC}\bigl(\frac{2}{\rc}\bigr)^3 uv\Bigr]u\pd{u}\\
  +\ic\biggl[\frac{1}{uv}-\frac{\kC}{\ic}-\frac{1}{\ic}\sq{\frac{2}{\rc}}\frac{2}{\kC}\Bigl[2+\frac{1}{4}\frac{\Lambda}{3}\bigl(\rc^2+\rH\vrcb\bigr)\Bigr]\biggr]v\pd{v}
  +\mathcal{O}(v^2)\,,
\end{multline}
which allows us to verify all the statements of Prop.~\ref{prop:sds:redshift} \emph{explicitly}; see \cite{vs:thesis} for details.

This completes the proof of Proposition \ref{prop:sds:local}.

%% file: localise.tex
Let us first look at the local redshift effect near the cosmological horizon with arbitrary incoming energy flux.
Recall the notation introduced before Cor.~\ref{cor:decay}, and the vectorfield $N$ of Prop.~\ref{prop:sds:redshift}. 

\begin{prop}\label{prop:localise:local}
Given $\psi$ in the class of finite energy solutions to \eqref{eq:sds:wave} as discussed in Thm.~\ref{thm:sds}, let us define an incoming energy flux density $g$ by
\begin{equation}
  \label{eq:localise:def:g}
    \int_{\mathcal{C}^+_{\tau_1,\tau_2}}{}^\ast J^N[\psi]=\int_{\tau_1}^{\tau_2}g(\tau)\ud \tau\,,
\end{equation}
where $\mathcal{C}^+_{\tau_1,\tau_2}=\mathcal{C}^+_{\tau_1}\setminus\mathcal{C}^+_{\tau_2}$, for any $\tau_2>\tau_1$. (Note that the left hand side is finite for any $\tau_2>\tau_1$ by Prop.~\ref{prop:r:static}.)
Let $r_0>\rc$ be chosen according to Prop.~\ref{prop:sds:redshift}, and define
\begin{equation}
  \label{eq:localise:def:f}
  f(\tau)=\int_{\mathrm{C}_\tau^\prime}{}^\ast J^N[\psi]\,,
\end{equation}
where $\mathrm{C}^\prime_{\tau}=\mathrm{C}_\tau\cap J^-(\Sigma_{r_0})$ denotes the segment of the outgoing null hypersurface $\mathrm{C}_{\tau}$ in the past of $\Sigma_{r_0}$ (c.f.~figure~\ref{fig:localise}). Then
\begin{equation}
  \label{eq:localise:gronwall}
  f(\tau_2)\leq f(\tau_1)e^{-b(\tau_2-\tau_1)}+\int_{\tau_1}^{\tau_2}g(\tau)e^{-b(\tau_2-\tau)}\ud\tau\,,
\end{equation}
for all $\tau_2>\tau_1$, where $b>0$ is a constant that only depends on $\Lambda$, and $m$.
\end{prop}

\begin{figure}
\includegraphics{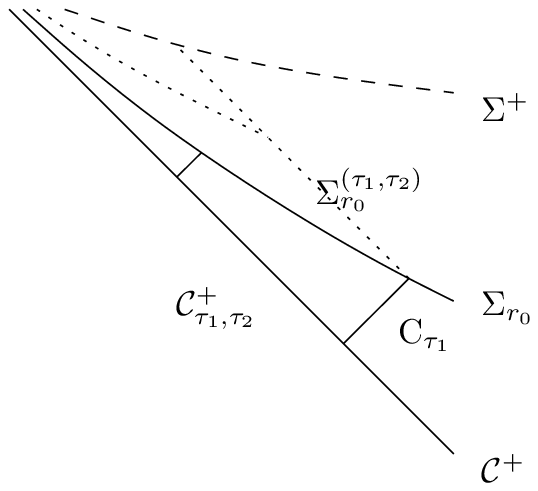}
\caption{Localisation of the global and local redshift effect.}
\label{fig:localise}
\end{figure}

\begin{rmk}
  If we have no incoming energy flux from the static region, $g=0$, then the energy $f$ decays exponentially and proportionally to the strength of the redshift effect as reflected in the constant $b>0$. In general, $g\geq 0$, the decay of $f$ cannot be stronger than that of the incoming energy flux. Note here however, that the second term in the upper bound is a convolution around $\tau=\tau_2$, and thus \emph{localises} the contribution from the incoming energy.
\end{rmk}

\begin{proof}
Consider the energy identity for $J^N$ on the domain bounded in the past by $\mathcal{C}^+_{\tau_1}\cup\mathrm{C}_{\tau_1}$ and in the future by $\mathrm{C}_{\tau_2}$ and $\Sigma_{r_0}$, $\tau_2>\tau_1$, namely
\begin{equation*}
  \label{eq:localise:energyidentity:full}
  \int_{\mathrm{C}^\prime_{\tau_2}}{}^\ast J^N +\int_{\Sigma_{r_0}^{(\tau_1,\tau_2)}}J^N\cdot n\:\dm{\gb{r_0}}+\int_{\mathcal{R}_{\tau_1,\tau_2}^{(r_0)}} K^N\dm{g}
  \leq\int_{\mathcal{C}^+_{\tau_1,\tau_2}}{}^\ast J^N +\int_{\mathrm{C}^\prime_{\tau_1}}{}^\ast J^N\,,
\end{equation*}
where $\Sigma_{r_0}^{(\tau_1,\tau_2)}=\Sigma_{r_0}\cap J^+(\mathrm{C}_{\tau_1})\cap J^-(\mathrm{C}_{\tau_2})$, and $\mathcal{R}_{\tau_1,\tau_2}^{(r_0)}=J^+(\mathcal{C}^+_{\tau_1}\cup\mathrm{C}_{\tau_1})\cap J^-(\mathrm{C}^\prime_{\tau_2}\cup\Sigma_{r_0}^{(\tau_1,\tau_2)})$.
Observe that in the double null coordinates of Section~\ref{sec:sds:dn},
\begin{equation*}
    \dm{g}=\frac{\Omega^2}{2}\,r^2\:\ud u\wedge\ud v\wedge\dm{\g}\,,
\end{equation*}
where $\Omega^2$ on $\mathcal{C}^+$ is a constant that merely depends on $\Lambda$ and $m$, 
\begin{equation*} 
  \Omega^2\onc=\frac{4}{\iota_\mathcal{C}\kC}\,,\qquad\iota_{\mathcal{C}}=\frac{2}{(\rc-\rH)^{\frac{\rH}{\rc}\frac{\rc+\vrcb}{\rH+\vrcb}}(\rc+\vrcb)^{\frac{\vrcb}{\rc}\frac{\rc-\rH}{\vrcb+\rH}}}\,.
\end{equation*}
Moreover, in view of \eqref{eq:sds:local:nullframe:horizon} and using \eqref{eq:sds:redshift:div},
\begin{equation*}
  K^N\onc\geq\min\Bigl\{\frac{1}{4}\kC\ic,2\bigl(\frac{2}{\rc}\bigr)^2\frac{\ic}{\kC}\Bigr\}\:\frac{1}{u}\,J^N\cdot \pd{v}\onc\,.
\end{equation*}
Therefore there exists a constant $b$ that only depends on $\Lambda$, and $m$, such that
\begin{equation*}
  \int_{\mathcal{R}_{\tau_1,\tau_2}^{(r_0)}}K^N\dm{g}\geq b(\Lambda,m)\int_{\tau_1}^{\tau_2}\ud\tau\int_{\mathrm{C}_{\tau}^\prime}{}^\ast J^N\,,
\end{equation*}
where we have used that by definition
\begin{equation*}
  \dd{u}{\tau}\onc=\frac{1}{\kC}\frac{1}{u}\,.
\end{equation*}
Thus the energy identity  takes the form
\begin{equation*}
    f(\tau_2)+b \int_{\tau_1}^{\tau_2}f(\tau)\, \ud\tau \leq
  \int_{\tau_1}^{\tau_2}g(\tau)\, \ud\tau +f(\tau_1)\,,\qquad(\tau_2>\tau_1)\,.
\end{equation*}
This is a Gronwall-type inequality which easily implies \eqref{eq:localise:gronwall}. Indeed, if we set
\begin{equation*}
  F(\tau_2,\tau_1)=\int_{\tau_1}^{\tau_2}\bigl\{b f(\tau)-g(\tau)\bigr\}\ud \tau\,,
\end{equation*}
then the inequality simply reads
\begin{equation*}
    f(\tau_2)+F(\tau_2,\tau_1)\leq f(\tau_1)\,.
\end{equation*}
Furthermore,
\begin{equation*}
  -\pd{\tau}\Bigl[F(\tau_2,\tau)e^{-b(\tau_2-\tau)}\Bigr]\geq-b\bigl[-f(\tau_2)\bigr]e^{-b(\tau_2-\tau)}-g(\tau)e^{-b(\tau_2-\tau)}\,,
\end{equation*}
so by integration,
\begin{equation*}
  F(\tau_2,\tau_1)e^{-b(\tau_2-\tau_1)}\geq f(\tau_2)\bigl(1-e^{-b(\tau_2-\tau_1)}\bigr)-\int_{\tau_1}^{\tau_2}g(\tau)e^{-b(\tau_2-\tau)}\ud\tau
\end{equation*}
which yields \eqref{eq:localise:gronwall} after rearranging and using the above inequality once again.
\end{proof}

Let us now complete the proof of Corollary~\ref{cor:decay}.

\begin{proof}[Proof of Cor.~\ref{cor:decay}]
In a first step we apply our Prop.~\ref{prop:localise:local} to the local redshift region.
Under the assumption \eqref{eq:cor:decay:assumption} we are in the situation that the incoming energy flux density is bounded from above by
\begin{equation*}
  g(\tau)=\frac{k}{\tau^{k+1}}\qquad(k\in\mathbb{N})\,,
\end{equation*}
and similary under the assumption of exponential decay by
\begin{equation*}
  g(\tau)=\beta e^{-\beta\tau}\qquad(\beta>0)\,.
\end{equation*}
Now we have, on one hand,
\begin{multline*}
  \int_{\tau_1}^{\tau_2}\frac{k}{\tau^{k+1}}e^{-b(\tau_2-\tau)}\ud\tau\leq
  \int_{\tau_1^\prime}^\infty\frac{k}{\tau^{k+1}}\ud \tau+e^{-b(\tau_2-\tau_1^\prime)}\int_{\tau_1}^{\tau_1^\prime}\frac{k}{\tau^{k+1}}\ud\tau\leq\\
  \leq\frac{1}{{\tau_1^\prime}^k}+e^{-\alpha\tau_2}\frac{1}{{\tau_1}^k}
  =\frac{1}{(1-\frac{\alpha}{b})^k}\frac{1}{\tau_2^k}+e^{-\alpha\tau_2}\frac{1}{\tau_1^k}\,,
\end{multline*}
where we have chosen $\tau_1^\prime\in(\tau_1,\tau_2)$ such that $\tau_2-\tau_1^\prime=\frac{1}{b}\alpha\,\tau_2$ for some $\alpha\in(0,b)$,
and on the other hand,
\begin{equation*}
  \int_{\tau_1}^{\tau_2}\beta e^{-\beta\tau} e^{-b(\tau_2-\tau)}\ud \tau\leq
  \begin{cases}
    \frac{\beta}{b-\beta}e^{-\beta\tau_2}&,\,\beta<b\\
    \beta (\tau_2-\tau_1)e^{-b\tau_2}&,\,\beta=b\\
    \frac{\beta}{\beta-b}e^{-b(\tau_2-\tau_1)}&,\,\beta>b
  \end{cases}\,.
\end{equation*}
We conclude by Prop.~\ref{prop:localise:local} if \eqref{eq:cor:decay:assumption} holds then also
\begin{equation*}
  \int_{\mathrm{C}_{\tau_2}^\prime}{}^\ast J^N[\psi]\leq\frac{C}{\tau_2^k}+e^{-b(\tau_2-\tau_1)}\int_{\mathrm{C}_{\tau_1}^\prime}{}^\ast J^N[\psi]\qquad(\tau_2>\tau_1)\,,
\end{equation*}
and similarly if \eqref{eq:cor:decay:assumption:exp} holds,
then there exists a constant $C$ (we assume for simplicity $\beta\neq b$) such that also
\begin{equation*}
  \int_{\mathrm{C}_{\tau_2}^\prime}{}^\ast J^N[\psi]\leq 
  C\:e^{-\min\{\beta,b\}\,\tau_2}+e^{-b(\tau_2-\tau_1)}\int_{\mathrm{C}_{\tau_1}^\prime}{}^\ast J^N[\psi]\qquad(\tau_2>\tau_1)\,.
\end{equation*}

In a second step we apply our Prop.~\ref{prop:sds:globalredshift} to the global redshift region. Here we \emph{localise} the argument of Section~\ref{sec:sds:globalredshift} to the domain of dependence of $\mathcal{C}^+_\tau\cup\mathrm{C}^\prime_\tau$ (as indicated in figure~\ref{fig:localise}).
Let $r_0$ be chosen as above, and denote for $r\geq r_0$ by
\begin{equation*}
  \Sigma_{r}^\prime = \bigl(\Sigma_r\cap J^+(\mathrm{C}_\tau) \bigr) \setminus J^+(\mathrm{C}_\tau\setminus \mathrm{C}_\tau^\prime)
\end{equation*}
the segment of $\Sigma_r$ in the domain of dependence of $\mathcal{C}_\tau^+\cup\mathrm{C}^\prime_\tau$ for any given $\tau>\tau_0$.
Then, as discussed in Section~\ref{sec:sds:globalredshift}, c.f.~in particular \eqref{eq:sds:global:weighted}, 
\begin{equation*}  
  r_2\int_{\Sigma_{r_2}^\prime}J^M\cdot n\,\dm{\gb{r_2}}\leq r_1\int_{\Sigma_{r_1}^\prime}J^M\cdot n\,\dm{\gb{r_1}}\qquad(r_2>r_1\geq r_0)\,;
\end{equation*}
note here that the vectorfield $M$ is \emph{timelike} and hence gives rise to a positive boundary term on the null segment of the future boudary of the domain of dependence of $\mathcal{C}^+_\tau\cup\mathrm{C}^\prime_\tau$, which preserves the inequality \eqref{eq:sds:global:coarea} with $\Sigma_r$ replaced by $\Sigma_r^\prime$. Next we choose $r_1=r_0$ and consider the energy identity for $J^N$ on the domain bounded in the past by $\mathcal{C}^+_\tau\cup\mathrm{C}^\prime_\tau$ and in the future by $\Sigma_{r_0}$, as in the proof of Prop.~\ref{prop:localise:local}. By Prop.~\ref{prop:sds:redshift} we then have
\begin{equation*}
  r_0\int_{\Sigma_{r_0}^\prime}J^M\cdot n\,\dm{\gb{r_o}}\leq
  C\int_{\mathcal{C}^+_\tau\cup\mathrm{C}^\prime_\tau}{}^\ast J^N\,,
\end{equation*}
which \emph{decays} according to our assumption and what is shown above.
We may now take the limit $r_2\to\infty$ as discussed in Section~\ref{sec:r:integral} to obtain the statement of the Corollary with $\tau$ replaced by $\tau^\prime$, where $\Sigma_{\tau^\prime}^+=(\Sigma^+\cap J^+(\mathrm{C}_\tau))\setminus J^+(\mathrm{C}_\tau\setminus\mathrm{C}_\tau^\prime)$. 
It remains to show that $\tau^\prime$ and $\tau$ are related linearly, which is an immediate consequence of geometric relations described in Section~\ref{sec:sds:dn}. For given values of $u$ and $r_0$ the incoming null hypersurface originating from the sphere at $(u,r_0)$ is
\begin{equation*}
  v=\frac{h_0}{u}\,,\qquad h_0=\frac{r_0-\rc}{(r_0-\rH)^{\frac{\rH}{\rc}\frac{\rc+\vrcb}{\rH+\vrcb}}(r_0+\vrcb)^{\frac{\vrcb}{\rc}\frac{\rc-\rH}{\vrcb+\rH}}}\,,
\end{equation*}
which terminates on $\Sigma^+$ at the sphere $(u^\prime,v)$, where
\begin{equation*}
  u^\prime=\frac{1}{v}=\frac{1}{h_0}\:u\,.
\end{equation*}
Therefore, since $T\cdot\tau=1$, also
\begin{equation*}
  \tau^\prime=\tau+\frac{1}{\kC}\int_{u}^{u^\prime}\frac{1}{u}\:\ud u=
  \tau+\frac{1}{\kC}\log\frac{1}{h_0}\,.\qedhere
\end{equation*}

\end{proof}

%% file: appl-gen.tex
In this Section we demonstrate that the decay mechanism manifested in the global redshift effect is sufficiently stable to extend to a general class of expanding spacetimes close to the Schwarzschild de Sitter cosmology.

The level of generality is comparable to the boundedness result of Dafermos and Rodnianski for the wave equation on a class of backgrounds close to the Schwarzschild spacetime \cite{dr:kerr:bounded}. Here however, \emph{no symmetries} are required in the expanding region, and \emph{no smallness assumption on the angular momentum} of the black hole is made (as long as it remains subextremal).

We consider a general class of $\mathrm{C}^1$-metrics with the global causal geometry of a Schwarzschild de Sitter metric and with the property that  the deformation tensor of the global redshift vectorfield decays sufficiently fast towards the future boundary.

\begin{defn}\label{def:appl:class}
Let $\Lambda>0$, and $0<3m\sqrt{\Lambda}<1$ be fixed, and denote the corresponding Schwarzschild de Sitter metric simply by $g_m$. Recall $g_m$ is here defined on a domain $\mathcal{D}=\mathcal{S}\cup\mathcal{C}\cup\mathcal{R}\cup\mathcal{\bar{C}}\cup\mathcal{S}$ covered by a single chart discussed in Section~\ref{sec:sds:dn}. Moreover in Section~\ref{sec:sds:foliation} we introduced an additional chart for the expanding region $\mathcal{R}^+$, (where $(\theta,\phi)$ are polar coordinates on the sphere):
\begin{equation}
  \label{eq:appl:chart:R}
  \Bigl\{(-\infty,\infty)\times(\rc,\infty)\times(0,\pi)\times(0,2\pi)\:,(t,r,\theta,\phi)\Bigr\}
\end{equation}
Let $\delta>0$ and $g$ be a  Lorentzian $\mathrm{C}^1$-metric defined on the fixed differentiable structure of $(\mathcal{D},g_m)$. We say {\boldmath{$g\in\mathcal{G}^{(\delta)}$}} if
\begin{enumerate}
\item  $g$ converges uniformly in a $\mathrm{C}^0$-sense to $g_m$ as $\Sigma^+$ is approached, such that
  \begin{gather*}
    \frac{\lvert g_{\mu\nu}-(g_m)_{\mu\nu}\rvert}{\sqrt{\lvert (g_m)_{\mu\mu}\rvert\lvert (g_m)_{\nu\nu}\rvert}}=\mathcal{O}\Bigl(\frac{1}{r^\delta}\Bigr)\displaybreak[0]\\
    \frac{\lvert (g^{-1})^{\mu\nu}-(g_m^{-1})^{\mu\nu}\rvert}{\sqrt{\lvert (g_m^{-1})^{\mu\mu}\rvert\lvert (g_m^{-1})^{\nu\nu}\rvert}}=\mathcal{O}\Bigl(\frac{1}{r^\delta}\Bigr)
  \end{gather*}
  with respect to the coordinate chart \eqref{eq:appl:chart:R}, \footnote{The condition as stated simply means that $g-g_m$ evaluated on the \emph{unit frame} $(E_0,E_1,E_2,E_3)$ --- with $E_0$, $E_1$, $E_2$, $E_3$ colinear to $\partial_r$, $\partial_t$, $\partial_\theta$, $\partial_\phi$, respectively, and unit w.r.t to $g_m$ ---  has the prescribed asymptotics.}
\item $g$ converges uniformly as a $\mathrm{C}^1$-metric to $g_m$ as $\Sigma^+$ is approached in the sense that the deformation tensor of the vectorfield 
  \begin{equation*}
    M^\prime=\frac{\partial}{\partial r}\,,\qquad {}^\prime\pi=\frac{1}{2}\mathcal{L}_{M^\prime} g\,,
  \end{equation*}
  decays sufficiently fast to its counterpart with respect to $g_m$, namely:
  \begin{gather*}
    {}^\prime\pi_m=\frac{1}{2}\mathcal{L}_{M^\prime}g_m\\
    \frac{\lvert {}^\prime\pi_{\mu\nu}-({}^\prime\pi_m)_{\mu\nu}\rvert}{\sqrt{\lvert(g_m)_{\mu\mu}\rvert\lvert(g_m)_{\nu\nu}\rvert}}=\mathcal{O}\Bigl(\frac{1}{r^{1+\delta}}\Bigr)\,,
  \end{gather*}
  with respect to the chart \eqref{eq:appl:chart:R}.

  \item $\mathcal{C}^+$ and $\bar{\mathcal{C}}^+$ are null hypersurfaces for $g$ which are generated by a Killing vectorfield $V$ with
    \begin{equation*}
      \nabla_V V=\kappa \,V\,,\qquad\kappa>0\,;
    \end{equation*}
i.e.~$\mathcal{C}^+$ and $\bar{\mathcal{C}}^+$ are Killing horizons with positive surface gravity. 
Moreover $V$ extends smoothly to a complete vectorfield $V$ on a neighborhood of $\mathcal{C}^+\cup\bar{\mathcal{C}}^+$ such that the flow along its integral curves has the following uniformity property (c.f.~Epilogue in \cite{dr:c}):
Let $\varphi_t$ be the 1-parameter group of diffeomorphisms generated by $V$, and $\Sigma=\Sigma_1\cup\Sigma_2$ be compact spacelike hypersurface segments such that $\Sigma_1\cap\mathcal{C}^+\neq\emptyset$, $\Sigma_2\cap\bar{\mathcal{C}}^+\neq\emptyset$, and $\min_{\Sigma_i}r=\rc-2\epsilon$, $\max_{\Sigma_i}r=\rc+2\epsilon$, $i=1,2$, $\epsilon>0$; then for all $t\geq 0$,
\begin{equation*}
  [\rc-\epsilon,\rc+\epsilon]\subset\Bigl(\inf_{\Sigma_t}r,\sup_{\Sigma_t}r\Bigr)\,,\quad\text{where }\Sigma_t=\varphi_t\bigl(\Sigma\bigr)\,.
\end{equation*}

\item Finally $g$ is time-oriented in the sense that there exists a vectorfield $M_0$ to the future of $\mathcal{C}^+\cup\bar{\mathcal{C}}^+$ such that:
\begin{equation*}
  g(M^\prime,M_0)<0\qquad\text{: on }\mathcal{R}^+\,.
\end{equation*}
    \end{enumerate}
    We refer to (1) and (2) as \emph{asymptotic closeness} conditions, and (3) and as a \emph{nondegenerate horizon} condition.
\end{defn}

The single most important example of a family of elements in that class are the \emph{Kerr de Sitter} spacetimes \cite{carter:lecture}. These spacetimes are not spherically symmetric, but they do retain two Killing fields $T$, and $\Phi$, and the metric $g$ decomposes into a metric $f$ on the group orbits of $T$, $\Phi$, and a metric $h$ on the orthogonal surfaces:
\begin{equation}
  \label{eq:kerrdesitter}
  g=f+h
\end{equation}
In \emph{Boyer-Lindquist}-type coordinates $(t,r,\theta,\phi)$ where $T=\partial_t$ and $\Phi=\partial_\phi$ we have
\begin{equation}
  \label{eq:kerrdesitter:f}
  f=\sin^2\theta\,\frac{\Delta_\theta}{\rho^2}\,\Bigl(a\,\ud t-\frac{r^2+a^2}{\Delta_0}\,\ud\phi\Bigr)^2-\frac{\Delta_r}{\rho^2}\Bigl(\ud t-\frac{a\sin^2\theta}{\Delta_0}\,\ud\phi\Bigr)^2\,,
\end{equation}
and
\begin{equation}
  \label{eq:kerrdesitter:h}
  h=\rho^2\Bigl[\frac{1}{\Delta_r}\,\ud r^2+\frac{1}{\Delta_\theta}\,\ud\theta^2\Bigr]\,,
\end{equation}
where we adopt the conventional notation,
\begin{subequations}\label{eq:kerrdesitter:notation}
  \begin{gather}
   \Delta_r=\bigl(r^2+a^2\bigr)\Bigl(1-\frac{\Lambda}{3}r^2\Bigr)-2mr\\
  \Delta_\theta=1+\frac{\Lambda}{3}\,a^2\,\cos^2\theta\\
  \Delta_0=1+\frac{\Lambda}{3}\,a^2\\
  \rho^2=r^2+a^2\cos^2\theta\,.
\end{gather}
\end{subequations}
Here $\lvert a\rvert>0$ quantifies the deviation from spherical symmetry, or physically speaking the rotation of the black hole. Note in particular that \eqref{eq:kerrdesitter} reduces to the Schwarzschild de Sitter metric in the case where $a=0$:
\begin{equation}
  \label{eq:appl:kerr:schwarzschild}
  g\rvert_{a=0}=g_m\,.
\end{equation}

\begin{prop}\label{prop:kerr:G}
  The members of the Kerr de Sitter family $(\mathcal{D},g_a)$ with $a$ in the subextremal range are contained in the general class of expanding cosmologies specified in Definition~\ref{def:appl:class}. In fact,
  \begin{equation*}
    g_a\in\mathcal{G}^{(2)}\,.
  \end{equation*}
\end{prop}

\begin{rmk}
  In so far as the rate of convergence is concerned the general class $\mathcal{G}^{(\delta)}$, $\delta>0$, is considerably larger than would be needed to encompass the perturbations introduced by slowly rotating Kerr de Sitter cosmologies. Indeed, Prop.~\ref{prop:kerr:G} states in particular that with respect to the chart \eqref{eq:appl:chart:R},
  \begin{equation}
    \frac{\lvert ({}^\prime\pi_a)_{\alpha\beta}-({}^\prime\pi_m)_{\alpha\beta}\rvert}{\sqrt{\lvert (g_m)_{\alpha\alpha}\rvert\lvert (g_m)_{\beta\beta}\rvert}}=\mathcal{O}\Bigl(r^{-3}\Bigr)\,.
  \end{equation}
  Moreover, $g_a$ is endowed with two Killing vectorfields which are not required in the expanding region $\mathcal{R}^+$ for $g_a$ to fall into the class $\mathcal{G}^{(2)}$.
\end{rmk}

The closeness property of the Kerr de Sitter family to the Schwarzschild de Sitter cosmology with the same value of the mass $m$ asserted in Prop.~\ref{prop:kerr:G} is \emph{not} immediately apparent from the expressions \eqref{eq:kerrdesitter:f} and \eqref{eq:kerrdesitter:h}. In fact, while \eqref{eq:appl:kerr:schwarzschild} is true, we have for example
\begin{equation}
  f_{t\phi}=-\frac{a}{\Delta_0}\sin^2\theta\Bigl[\frac{\Lambda}{3}\bigl(r^2+a^2\bigr)+\frac{2mr}{\rho^2}\Bigr]
\end{equation}
in Boyer-Lindquist coordinates, which \emph{diverges} for arbitrary small parameter $\lvert a\rvert>0$ as $r\to\infty$ from its counterpart $(g_m)_{t\phi}=0$ in the Schwarzschild de Sitter geometry.
In order to capture the stated closeness property of Kerr de Sitter we pass to a ``co-rotating'' coordinate system $(t_0,\phi_0,\theta_0,\phi_0)$, where in particular
\begin{equation}
  \phi_0(t,\phi)=\phi-\frac{\Lambda}{3}a\,t\,.
\end{equation}
The relevant transformation already appears in \cite{carter:lecture, hawking:rot:ads}; see also \cite{teitelboim:ads}. So as not to distract from the main argument in this section we defer the proof of the asymptotic closeness conditions of Prop.~\ref{prop:kerr:G} to Appendix~\ref{a:kerr}.

The nondegenerate horizon condition is a standard property of the (cosmological) horizons of \emph{subextremal} Kerr de Sitter spacetimes. 
\smallskip

The purpose of the following proposition is to show that the assumptions on the class $\mathcal{G}^{(\delta)}$ are sufficiently restrictive for the global redshift effect to come into play. Since for any given $g\in\mathcal{G}^{(\delta)}$ the differentiable structure of $(\mathcal{D},g)$ contains the chart \eqref{eq:appl:chart:R} we can adhere to our choice of the global redshift vectorfield with respect to the coordinates $(t,r,\phi,\theta)$:
\begin{equation}
  M=\frac{1}{r}\frac{\partial}{\partial r}
\end{equation}
We continue to denote by $\phi$ the lapse function of the foliation of $\mathcal{R}^+$ by the level sets $\Sigma_r$ of the coordinate function $r$ with respect to the metric $g$, and by $n$ the normal to $\Sigma_r$ with respect to $g$. In formulas:
\begin{equation}
  n=\phi\, V\,,\qquad V^\mu=g^{\mu\nu}\partial_\nu r\,,\quad \phi^{-2}=-g(V,V)\,.
\end{equation}

\begin{prop}\label{prop:appl:globalredshift}
  Let $\psi$ be a solution to the wave equation on $(\mathcal{D},g)$, where $g\in \mathcal{G}^{(\delta)}$ for some $\delta>0$. Then on $\mathcal{R}^+$:
  \begin{equation*}
    \phi\:\nabla\cdot J^M\geq\frac{1}{r} J^M\cdot n-\mathcal{O}\Bigl(\frac{1}{r^{1+\delta}}\Bigr)J^M\cdot n\,.
  \end{equation*}  
\end{prop}

\begin{proof}
  We denote by $T$, $\pi$, ${}^\prime\pi$, $K^M$, $K^{M^\prime}$, the energy-momentum tensor \eqref{eq:energymomentum}, the deformation tensor of $M$, that of $M^\prime=r M$, and the currents \eqref{eq:KMd} of $M$, and $M^\prime$, \emph{with respect to $g$}, respectively, and by $T_m$, ${}^\prime\pi_m$, $\pi_m$, $K_m^M$, $K_m^{M^\prime}$ the corresponding quantities with respect to $g_m$.

It follows immediately from Def.~\ref{def:appl:class}~(1) that:
\begin{gather*}
  \lvert g^{\alpha\beta}\partial_\alpha\psi\partial_\beta\psi- g_m^{\alpha\beta}\partial_\alpha\psi\partial_\beta\psi\rvert=\OdO\sum_{\alpha}\lvert g_m^{\alpha\alpha}\rvert(\partial_\alpha\psi)^2\\
  \lvert T_{\mu\nu}-(T_m)_{\mu\nu}\rvert=\OdO\Bigl[\lvert(g_m)_{\mu\nu}\rvert+\sqrt{\lvert(g_m)_{\mu\mu}\rvert\lvert(g_m)_{\nu\nu}\rvert}\Bigr]\sum_\alpha\lvert g_m^{\alpha\alpha}\rvert\bigl(\partial_\alpha\psi\bigr)^2
\end{gather*}
We verify readily using the explicit expressions for $\pirm^{tt}$, $\pirm^{rr}$,$\pirm^{AB}$ already calculated in the proof of Prop.~\ref{prop:sds:globalredshift} that
\begin{equation*}
    \sqrt{\lvert(g_m)_{\mu\mu}\rvert\lvert(g_m)_{\nu\nu}\rvert}\lvert\lvert\pirm^{\mu\nu}\rvert=\lvert(g_m)_{\mu\mu}\rvert\lvert\pirm^{\mu\mu}\rvert=\frac{4}{r}+\mathcal{O}\Bigl(\frac{1}{r^3}\Bigr)\,.
\end{equation*}
Therefore
\begin{equation*}
  \lvert g_m^{\mu\alpha}g_m^{\nu\beta}\pirm_{\alpha\beta}(T_{\mu\nu}-(T_m)_{\mu\nu})\rvert=\Od{1}\sum_\alpha\lvert g_m^{\alpha\alpha}\rvert\bigl(\partial_\alpha\psi\bigr)^2\,.
\end{equation*}
Now by Def.~\ref{def:appl:class}~(2) this is matched by
\begin{equation*}
  \lvert g_m^{\mu\alpha}g_m^{\nu\beta}({}^\prime\pi_{\alpha\beta}-\pirm_{\alpha\beta})(T_m)_{\mu\nu}\rvert=\Od{1}\sum_\alpha\lvert g_m^{\alpha\alpha}\rvert\bigl(\partial_\alpha\psi\bigr)^2\,,
\end{equation*}
and thus
\begin{equation*}
  g_m^{\mu\alpha}g_m^{\nu\beta}\,{}^\prime\pi_{\alpha\beta}\,T_{\mu\nu}
  =g_m^{\mu\alpha}g_m^{\nu\beta}\pirm_{\alpha\beta}(T_m)_{\mu\nu}+\Od{1}\sum_\alpha\lvert g_m^{\alpha\alpha}\rvert\bigl(\partial_\alpha\psi\bigr)^2\,.
\end{equation*}
Since
\begin{equation*}
  \lvert g_m^{rr}\rvert (T_m)_{rr}=\frac{1}{2}\sum\lvert g_m^{\alpha\alpha}\rvert\bigl(\partial_\alpha\psi)^2\,,
\end{equation*}
we obtain finally
\begin{equation*}
  K^{M^\prime}=g^{\mu\alpha}g^{\nu\beta}\,{}^\prime\pi_{\alpha\beta}\,T_{\mu\nu}=K_m^{M^\prime}+\Od{1}\frac{1}{\phi_m^2}(T_m)_{rr}\,,
\end{equation*}
where $\phi_m$ denotes the lapse function \eqref{eq:sds:lapse}.
The crucial inequality in the proof of the global redshift proposition~\ref{prop:sds:globalredshift} is
\begin{equation*}
  K_m^{M^\prime}\geq 0\,,
\end{equation*}
which yields in the present setting:
\begin{equation*}
  K^{M^\prime}\geq -\Od{1}\frac{1}{\phi_m^2}(T_m)_{rr}\,.
\end{equation*}
Since, in view of the properties Def.~\ref{def:appl:class}~(1),
\begin{gather*}
  \pi_{\alpha\beta}=-\frac{1}{2}\frac{1}{r^2}(\partial_\alpha r)g_{r\beta}-\frac{1}{2}\frac{1}{r^2}(\partial_\beta r)g_{\alpha r_0}+\frac{1}{r}{}^\prime\pi_{\alpha\beta}\\
  K^M=-\frac{1}{r^2}g^{\mu r}g^{\nu\beta}g_{r\beta}T_{\mu\nu}+\frac{1}{r}K^{M^\prime}
  =-\frac{1}{r^2}g_m^{rr} T_{rr}+\frac{1}{r}K^{M^\prime}
  +\Od{2}\lvert g_m^{rr}\rvert T_{rr}\,,
\end{gather*}
we see that our bound on $K^{M^\prime}$ translates into the following bound on $K^M$:
\begin{equation}
  K^M\geq\frac{1}{r^2}\frac{1}{\phi_m^2}T_{rr}-\Od{2}\frac{1}{\phi_m^2}T_{rr}\,.
\end{equation}
Also for the normal $n$ we find in comparison to $n_m$:
\begin{gather*}
  V^\mu=V_m^\mu+\OdO\sqrt{ \lvert g_m^{\mu\mu}\rvert\lvert g_m^{rr}\rvert}\,,\\
  g(V,V)=g(V_m,V_m)+\OdO\lvert g_m^{rr}\rvert\,,\qquad
  \frac{1}{\phi^2}=\frac{1}{\phi_m^2}+\OdO\frac{1}{\phi_m^2}\,,\\
  n=\phi V=n_m+\OdO n_m+\OdO\sqrt{\lvert g_m^{\mu\mu}\rvert}\partial_\mu\,,
\end{gather*}
where
\begin{equation*}
  T(M,\sqrt{\lvert g_m^{\mu\mu}\rvert}\partial_\mu)\leq\mathcal{O}(1)\:T(M,n_m)\,.
\end{equation*}
So we have
\begin{equation}
  \phi \:K^M\geq\frac{1}{r}T(M,n_m)-\Od{1}T(M,n_m)
\end{equation}
which implies finally by the above
\begin{equation}
  \phi \:\nabla\cdot J^M\geq\frac{1}{r} J^M\cdot n-\Od{1}J^M\cdot n\,.\qedhere
\end{equation}
\end{proof}

The fact that the additional error term in Prop.~\ref{prop:appl:globalredshift} (c.f.~Prop.~\ref{prop:sds:globalredshift}) is \emph{integrable} allows us essentially to repeat the argument leading to Prop.~\ref{prop:r:expanding} in this larger class of expanding spacetimes $\mathcal{G}^{(\delta)}$. 
Indeed, since with respect to an orthonormal frame field\footnote{We may complement $n$ by an orthonormal frame $(e_1,e_2,e_3)$ tangential to $\Sigma_r$, and introduce a dual frame $\theta^\mu$ such that $g=g_{\alpha\beta}\theta^\alpha\otimes\theta^\beta=-\theta^0\otimes\theta^0+{(\gb{r})}_{ij}\theta^i\otimes\theta^j$ where $\theta^0=\phi\,\ud r$. Note that $r$ parametrizes the integral curves of $\phi\, n$, and that $\ud r$ is here dual to $\phi \,n$ as opposed to the coordinate vectorfield $\partial_r$ in the chart $(t,r,\theta,\phi)$.}  we have
\begin{equation}
  g=-\phi^2\ud r^2+\gb{r}\,,
\end{equation}
where $\gb{r}$ denotes the first fundamental form of $\Sigma_r$ in $(\mathcal{D},g)$, we obtain from the energy identity for $J^M$ on the domain bounded in the past by $\Sigma_{r_1}$ and to the future by $\Sigma_{r_2}$, for any $r_2>r_1\gg\rc$,
\begin{multline}
  \int_{\Sigma_{r_2}}J^M\cdot n\:\dm{\gb{r_2}}-\int_{\Sigma_{r_1}}J^M\cdot n\:\dm{\gb{r_1}}
  =-\int_{r_1}^{r_2}\ud r\int_{\Sigma_r}\phi\,\nabla\cdot J^M\dm{\gb{r}}\\
  \leq-\int_{r_1}^{r_2}\ud r\frac{1}{r}\int_{\Sigma_r}J^M\cdot n\:\dm{\gb{r}}+\int_{r_1}^{r_2}\ud r\:\mathcal{O}\Bigl(\frac{1}{r^{1+\delta}}\Bigr)\int_{\Sigma_r}J^M\cdot n\:\dm{\gb{r}}\,,
\end{multline}
where we have used the coarea formula and applied Prop.~\ref{prop:appl:globalredshift}. Hence the energy flux through $\Sigma_r$,
\begin{equation}
  f(r)=\int_{\Sigma_r}J^M\cdot n\:\dm{\gb{r}}\,,
\end{equation}
or in fact the rescaled energy $r f(r)$ satisfies the differential inequality:
\begin{equation}
  \bigl[ r f(r)\bigr]^\prime= f(r)+r\:f^\prime(r)\leq \mathcal{O}\Bigl(\frac{1}{r^{1+\delta}}\Bigr) \bigl[r\,f\bigr]\,.
\end{equation}
The immediate application of Gronwall's Lemma yields
\begin{equation}
  r f(r)\leq r_+ f(r_+)\: \exp\Bigl[{C\int_{r_+}^r\frac{1}{{r^\prime}^{1+\delta}}\ud r^\prime}\Bigr]\qquad (r>r_+\gg \rc)\,,
\end{equation}
where $C$ is a constant that only depends on $(\Lambda,m)$.
In particular, for any $\delta>0$,
\begin{equation}
  \lim_{r\to\infty}\int_{\Sigma_r}r\:J^M\cdot n\:\dm{\gb{r}}\leq C(r_+, a, \delta, \Lambda, m)\: \int_{\Sigma_{r_+}} J^M\cdot n\:\dm{\gb{r_+}}
\end{equation}
and the left hand side captures the \emph{same} rescaled asymptotic energy already obtained in the unperturbed Schwarzschild de Sitter setting, c.f.~\eqref{eq:sds:overview:finiteintegral}, as $g$ converges to $g_m$ according to Def.~\ref{def:appl:class}~(1), and consequently $T$ and $n$ converge to its corresponding quantities, $T_m$, $n_m$ respectively.

Since we have a strictly timelike vectorfield $M_0$ on $\mathcal{R}^+$ for any $g\in\mathcal{G}^{(\delta)}$ there is a bounded positive function $B$ on $(\rc,\infty)$ such that
\begin{equation}
  \lvert \nabla\cdot J^{M_0}\rvert \leq B(r)\: J^{M_0}\cdot n\qquad\text{: on }\Sigma_r\:\text{for }r>\rc\,,
\end{equation}
so that the trivial bound from Gronwall's lemma gives:
\begin{equation}
  \int_{\Sigma_{r_+}}J^{M_0}\cdot n\,\dm{\gb{r_+}}\leq C(r_+,r_-,a)\int_{\Sigma_{r_-}}J^{M_0}\cdot n\,\dm{\gb{r_-}}\,,
\end{equation}
for any fixed $r_+>r_->\rc$, where $r_-$ can be chosen as close to the cosmological horizons $\mathcal{C}\cup\bar{\mathcal{C}}$ as desired, c.f.~Fig.~\ref{fig:appl:R}.

\begin{figure}[tb]
  \centering
  \includegraphics{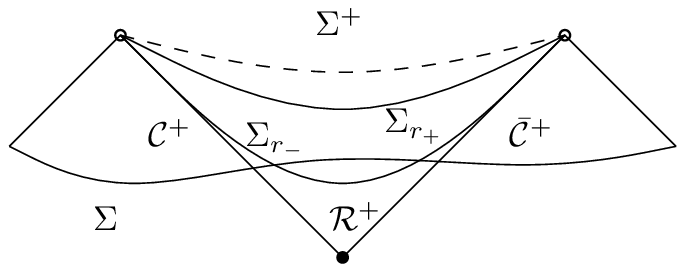}
  \caption{Global redshift argument on a general class of expanding spacetimes $\mathcal{G}_a^{(\delta)}$.}
  \label{fig:appl:R}
\end{figure}

We have thus divided the domain to the future of the cosmological horizons into a sufficiently fast expanding region $r>r_+$ where the global redshift effect dominates and an intermediate region $r_-\leq r \leq r_+$ where the energy is uniformly bounded due to the presence of timelike vectorfield on a domain of compact range in $r$. In the remaining region $\rc\leq r\leq r_+$ the behaviour of linear waves is controlled by the local redshift effect.

\begin{prop}
  Let $\Sigma$ be a spacelike hypersurface in $\mathcal{D}$ with $\mathcal{R}^+\cap\Sigma\ne\emptyset$, and denote by $\Sigma^\prime=\Sigma\cap\mathcal{R}^+$, $\mathcal{C}_0^+=\mathcal{C}\cap J^+(\Sigma)$, and $\Sigma_r^\prime=\Sigma_r\cap J^+(\Sigma)$. Given $g\in\mathcal{G}^{(\delta)}$ there exists a timelike vectorfield $N$ which is invariant under the flow generated by~$V$ and defined on a neighborhood of $\mathcal{C}_0^+$ which includes $\Sigma_{r_-}^\prime$ for $r_->\rc$ chosen sufficiently small, and a constant $C$ that only depends on $g$, such that
  \begin{equation*}
    \int_{\Sigma_{r_-}^\prime}J^M[\psi]\cdot n\:\dm{\gb{r_-}}\leq C\int_{\mathcal{C}_0^+} {}^\ast J^N[\psi]+C\int_{\Sigma^\prime}J^N[\psi]\cdot n\:\dm{\gb{}}
  \end{equation*}
for all solutions $\psi$ to the wave equation on $(\mathcal{D},g)$.
\end{prop}

This is an immediate application of Theorem~7.1 in the epilogue of \cite{dr:c}. The construction therein of the local redshift vectorfield $N$ applies to all non-extremal Killing horizons, in particular the cosmological horizons of $g\in\mathcal{G}^{(\delta)}$ under the assumption (3) of Definition~\ref{def:appl:class}.

\smallskip
The desired concluding estimate, however,
\begin{equation}
  \label{eq:appl:boundedness:S}
  \int_{\mathcal{C}_0^+}{}^\ast J^N[\psi]\leq \int_{\Sigma} J^n[\psi]\cdot n\:\dm{\gb{}}\,,
\end{equation}
\emph{cannot} be established in the generality of perturbations $g\in\mathcal{G}^{(\delta)}$. The reason is that for general perturbations of the Schwarzschild de Sitter metric in the static region $\mathcal{S}$ the previously causal vectorfield $T$ may be spacelike with respect to $g$ thus introducing an \emph{ergoregion} near the event horizon \cite{dr:c}. In \cite{dr:kerr:bounded} Dafermos and Rodnianski establish a uniform boundedness result analogous to \eqref{eq:appl:boundedness:S} for solutions to the wave equation on \emph{stationary axisymmetric} spacetimes which are $\mathrm{C}^1$-close to the Schwarzschild metric, (and whose Killing fields span the null generator of the event horizon). Their method which only relies on the presence of two Killing fields $T$, and $\Phi$ --- and in particular does not exploit any \emph{hidden} symmetries of the Kerr black holes --- can be adapted to the present setting as discussed in their concluding notes in \cite{dr:kerr:bounded}.

%% file: appl-kg.tex
The approach followed in Section~\ref{sec:proofs:global} lends itself to the study of solutions to the Klein-Gordon equation \eqref{eq:sds:kg},
\begin{equation}
  \label{eq:sds:kg:r}
  \Box_g\psi=\m\psi\,,
\end{equation}
i.e.~wave equations on Schwarzschild de Sitter backgrounds with a \emph{positive mass term} $\m>0$.

\begin{thm}\label{thm:sds:kg}
  Let $\Sigma$ be a $\Sigma^+$-hypersurface with normal $n$ in a (subextremal) Schwarzschild de Sitter spacetime $(\M,g)$; c.f.~figure~\ref{fig:cauchyexpanding}. Fix $\m>0$, then all solutions $\psi$ to \eqref{eq:sds:kg:r} with finite energy on $\Sigma$,
  \begin{equation*}
    D_{\m}[\psi]=\int_\Sigma\Bigl\{T[\psi](n,n)+\m\:\psi^2\Bigr\}\dm{\gb{}}<\infty\,,
  \end{equation*}
  are uniformly bounded on $\mathcal{R}^+$ in the norm
  \begin{equation*}
    \Vert\psi\Vert_{\Sigma_r}=\int_{\Sigma_r}\Bigl\{T(n,n)+\frac{1}{2}\m\psi^2\Bigr\}\frac{\phi}{r^2}\dm{\gb{r}}
  \end{equation*}
and have a limit on $\Sigma^+$ in $\mathrm{H}^1(\mathbb{R}\times\mathbb{S}^2)$. Moreover the limit $\psi\rvert_{\Sigma^+}$ is a square integrable function on the standard cylinder $\mathbb{R}\times\mathbb{S}^2$ and we have the bound:
  \begin{equation*}
    \int_{\Sigma^+}\psi^2\:\dm{\cg}\leq C\:D_{\m}[\psi]\,,
  \end{equation*}
  where $C$ is a constant that only depends on $m$, $\Lambda$, $\m$, and $\Sigma$.
\end{thm}

\begin{rmk}
We can commute \eqref{eq:sds:kg:r} with $T$ and $\Oi{i}=i:1,2,3$ to recover the bound of Thm.~\ref{thm:sds} for the homogeneous wave equation.
While as in the case $\m=0$ we obtain that the limit can be viewed as a function of finite energy on the standard cylinder, we obtain here for $\m>0$ that in addition $\psi\rvert_{\Sigma^+}\in\mathrm{L}^2(\mathbb{R}\times\mathbb{S}^2)$.
\end{rmk}

To prove Thm.~\ref{thm:sds:kg} we proceed as in Section~\ref{sec:r:integral} and Section~\ref{sec:proofs:global} for the proof of Thm.~\ref{thm:sds}. In other words, we show analogous versions of the local and global redshift effect for Klein-Gordon equations, which yield the monotone rescaled energy $\Vert\psi\Vert_{\Sigma_r}$ in the statement of theorem.

The energy currents used in Section~\ref{sec:proofs:global} can be adapted to incorporate the presence of a positive mass term. In fact, we merely replace \eqref{eq:standardcurrent} by the \emph{modified energy current}
\begin{equation}  
  \label{eq:modifiedcurrent}
  J^{M;\m}\doteq J^M-\frac{\m}{2}\psi^2\,M^\flat\,,
\end{equation}
where $M^\flat$ is simply the 1-form corresponding to the vectorfield $M$, i.e.~$M^\flat\cdot X=g(M,X)$. Here we retain \eqref{eq:energymomentum} as the definition of the energy momentum tensor. Then we obtain the following analogous statement to Prop.~\ref{prop:r:expanding}.

\begin{prop}[Global redshift property for Klein-Gordon equations]
  \label{prop:r:expanding:kg}
  Let $J^{M;\m}$ be the modified energy current \eqref{eq:modifiedcurrent} associated to the multiplier
  \begin{equation}
    \label{eq:M:kg}
    M=\frac{1}{r^3}\frac{\partial}{\partial r}\,,
  \end{equation}
  and a fixed mass $\m>0$. Then for all solutions $\psi$ to the Klein-Gordon equation \eqref{eq:sds:kg:r} we have
  \begin{equation*}
    \phi\:\nabla\cdot J^{M;\m}\geq\frac{1}{r} J^{M;\m}\cdot n\,,
  \end{equation*}
  on $\mathcal{R}^+$, i.e.~for $r>\rc$.
\end{prop}

Note that the global redshift vectorfield \eqref{eq:M:kg} carries additional weights in $r$ as compared to \eqref{eq:sds:M}.
\begin{proof}
  By virtue of \eqref{eq:sds:kg:r} we have
  \begin{equation*}
    \nabla^\mu J_\mu^{M;\m}=T_{\mu\nu}\deformt{M}^{\mu\nu}-\frac{1}{2}\m\tr\deformt{M}\:\psi^2
  \end{equation*}
  Recall our calculation of the deformation tensor of $M^\prime=\partial_r$ in the proof of Prop.~\ref{prop:sds:globalredshift}. Since, using the notation therein,
  \begin{equation*}
    \tr\deformt{M^\prime}=\Gamma_{rt}^t+\Gamma_{rr}^r+\frac{2}{r}=K_1-K_0+\frac{2}{r}=\frac{2}{r}\,,
  \end{equation*}
  we obtain for the vectorfield \eqref{eq:M:kg}:
  \begin{equation*}
    \tr\deformt{M}=-\frac{3}{r^4}+\frac{1}{r^3}\tr\deformt{M^\prime}=-\frac{1}{r^4}\,.
  \end{equation*}
  Therefore
  \begin{equation*}
    \nabla\cdot J^{M;\m}=-\frac{3}{r^4}g^{rr}T_{rr}+\frac{1}{r^3}K^{M^\prime}+\frac{\m}{2}\frac{1}{r^4}\psi^2
    \geq\frac{3}{r^4}\frac{1}{\phi^2}T_{rr}+\frac{\m}{2}\frac{1}{r^4}\psi^2\,,
  \end{equation*}
  because $K^{M^\prime}=T_{\mu\nu}\deformt{M^\prime}^{\mu\nu}\geq 0$ as shown in the proof of Prop.~\ref{prop:sds:globalredshift}.
Now,
\begin{equation*}
  J^{M;\m}\cdot n=T(M,n)-\frac{\m}{2}\psi^2\:g(M,n)=\frac{1}{\phi}\frac{1}{r^3}T_{rr}+\frac{\m}{2}\frac{1}{r^3}\phi\:\psi^2\,,
\end{equation*}
where we have used the relation $\partial_r=\phi n$.
This completes the proof of the global redshift property.
\end{proof}

\begin{rmk}
  It is clear that the global redshift property is stable under perturbations discussed in Section~\ref{sec:general}. Similarly to our generalisation of Prop.~\ref{prop:sds:globalredshift} to Prop.~\ref{prop:appl:globalredshift} we can show that Prop.~\ref{prop:r:expanding:kg} holds in a class of nearby cosmologies with an error that is lower order in $r$; c.f.~Prop.~\ref{prop:appl:globalredshift}.
\end{rmk}

As in (\ref{eq:M:redshift}--\ref{eq:sds:global:weighted}) we conclude from Prop.~\ref{prop:r:expanding:kg} that
\begin{equation}
  r_2\int_{\Sigma_{r_2}}J^{M;\m}\cdot n\,\dm{\gb{r_2}}\leq r_1\int_{\Sigma_{r_1}}J^{M;\m}\cdot n\,\dm{\gb{r_1}}\,,
\end{equation}
for all $r_2>r_1>\rc$, which is the uniform boundedness statement in Thm.~\ref{thm:sds:kg}.
Moreover in view of \eqref{eq:sds:overview:volumeform} we obtain in particular the finiteness of the limit:
\begin{equation}
  \lim_{r\to\infty}r\int_{\Sigma_r}J^{M;\m}\cdot n\:\dm{\gb{r}}=
  \int_{\Sigma^+}\Bigl\{T(n,n)+\frac{1}{2}\m\psi^2\Bigr\}\dm{\cg}
\end{equation}

Let us turn to the local redshift effect at the cosmological horizons.
We need to show that Prop.~\ref{prop:sds:redshift} remains valid in the presence of a mass $\m>0$.
\begin{prop}[Local redshift effect for Klein-Gordon equations]
  \label{prop:localredshift:kg}
  Let $N$ be constructed as in Section~\ref{sec:sds:localredshift}. Then for all solutions~$\psi$ to \eqref{eq:sds:kg:r} we have
  \begin{equation*}
    \nabla\cdot J^{N;\m}\geq b\: J^{N;\m}\cdot N\quad\text{: on }\mathcal{C}^+\,,
  \end{equation*}
  where $b$ is a constant that only depends on $m$, $\Lambda$.
\end{prop}

\begin{proof}
  It suffices to show that
  \begin{equation*}
    -\tr\deformt{N}> 0\,.
  \end{equation*}
  Recall the null frame $(T,Y,E_1,E_2)$ on $\mathcal{C}^+$. We have
  \begin{equation*}
    \tr\deformt{N}=-\deformt{N}(T,Y)+g^{AB}\deformt{N}_{AB}\,,
  \end{equation*}
  where $N=T+Y$.
  On one hand, on $\mathcal{C}^+$,
  \begin{equation*}
    \deformt{N}(T,Y)=\deformt{Y}(T,Y)=\frac{1}{2}g(\nabla_T Y,Y)+\frac{1}{2}g(T,\nabla_Y Y)=-\frac{1}{2}\sigma g(T,Y)=\sigma\,,
  \end{equation*}
  by construction, according to which $\nabla_Y Y=-\sigma N$.
  On the other hand, on $\mathcal{C}^+$,
  \begin{equation*}
    \nabla_{E_A}Y=\frac{2}{\rc}E_A
  \end{equation*}
  and thus
  \begin{equation*}
    \deformt{N}_{AB}=\frac{2}{\rc}g_{AB}\,.
  \end{equation*}
  Hence
  \begin{equation*}
    \tr\deformt{N}=-\sigma+\frac{4}{\rc}<0\,,
  \end{equation*}
  for $\sigma>0$ chosen large enough.
\end{proof}

We conclude the proof of the theorem with the required boundedness statement for the nondegenerate energy flux through the cosmological horizon.

\begin{prop}\label{prop:r:static:kg}
  Let $\Sigma$ and $\mathcal{C}^+_0$ be as in Prop.~\ref{prop:r:static}; c.f.~figure~\ref{fig:not:redshift:props}.
  Then for all solutions $\psi$ to \eqref{eq:sds:kg:r}, with $\m>0$ and $N$ as in Prop.~\ref{prop:localredshift:kg},
  \begin{equation*}
    \int_{\mathcal{C}^+_0}{}^\ast J^{N;\m}[\psi]\leq C \int_\Sigma J^{n;\m}[\psi]\cdot n\:\dm{\gb{}}\,,
  \end{equation*}
  where $C$ is a constant that only depends on $n$, $\Lambda$ and $\Sigma.$
\end{prop}

We remark that this statement is missing in the literature. A proof however can be obtained by revisiting the currents used in \cite{dr:sds}. Since the treatment of the static region is not at the centre of our interest in this paper, we only give a proof in the case $\m\geq 2\frac{\Lambda}{3}$, and $m=0$, that is for the Klein-Gordon equation on de Sitter (c.f.~Section~\ref{sec:ds}) where the discussion simplifies considerably.

\begin{proof}[Proof for $m=0$ and $\m\geq 2\frac{\Lambda}{3}$.]
  Apply the energy identity associated to the current $J^{N;\m}$ to the domain $J^+(\Sigma)\cap J^-(\mathcal{C}_0^+)$. The statement of the proposition follows provided
  \begin{equation*}
    \int_{J^+(\Sigma)\cap J^-(\mathcal{C}_0^+)} \nabla\cdot J^{N;\m}[\psi]\ \dm{g}\leq C\int_{\Sigma} J^{n;\m}[\psi]\cdot n\,\dm{\gb{}}\,.
  \end{equation*}
  In view of the definition and local redshift property of $N$, it sufficies to show that
  \begin{equation*}
    \int_{J^+(\Sigma)\cap J^-(\mathcal{C}_0^+)} \chi\: J^{n;\m}[\psi]\cdot n\:\dm{g}\leq C\int_{\Sigma} J^{n;\m}[\psi]\cdot n\,\dm{\gb{}}\,,
  \end{equation*}
  where $0\leq \chi\leq 1$ is function supported in a radial region $r_0\leq r\leq r_1<\rc$ away from the cosmological horizon.
  This is achieved with the current
  \begin{equation*}
    J^{X,1;\m}=J^{X,1}-\frac{\m}{2}X^\flat \psi^2
  \end{equation*}
  where, $X$ is a Morawetz vectorfield, c.f.~\eqref{eq:sds:tortoise},
  \begin{gather*}
    X=f(r)\frac{\partial}{\partial r^\ast}\,,\\
    J^{X,1}_\mu=J^X_\mu+\frac{1}{4}\:\bigl( f^\prime +\frac{2}{r} f\bigr)\bigl(1-\frac{\Lambda r^2}{3}\bigr)\:\partial_\mu\bigl(\psi^2\bigr)-\frac{1}{4}\:\partial_\mu\Bigl[ \bigl( f^\prime +\frac{2}{r} f\bigr)\bigl(1-\frac{\Lambda r^2}{3}\bigr)\Bigr]\:\psi^2\,.
  \end{gather*}
  Then with the simple choice $f=1$ we have for all solutions to \eqref{eq:sds:kg:r},
  \begin{equation*}
    \nabla\cdot J^{X,1;\m}=\frac{1}{r}\sqv{\nablab\psi}+\frac{1}{r}\frac{\Lambda r^2}{3}\Bigl(\m-2\frac{\Lambda}{3}\Bigr)\psi^2\geq 0\,.
  \end{equation*}
  The boundary terms of this current are controlled with the conserved energy associated to the vectorfield $T$.
   The remaining derivatives can be obtained with auxillary currents. For example choosing $f=r^2$ yields
   \begin{equation*}
     \nabla\cdot J^{X;\m}=2r\sq{\frac{\partial\psi}{\partial t}}+\bigl(2\frac{\Lambda r^2}{3}-1\bigr)r\sqv{\nablab\psi}+\m r \bigl(3\frac{\Lambda r^2}{3} -2\bigr)\psi^2\,,
   \end{equation*}
   which allows us to retrieve the $T$ derivative; similarly for the radial derivative.
\end{proof}

%% file: acoercivity.tex
In this appendix we shall prove the coercivity formula for the standard sphere, and recall the classic Sobolev inequality on the sphere.

Let us denote by
\begin{equation}
  S_r=\Bigl\{x\in\mathbb{R}^3:\quad\lvert x\rvert=r\Bigr\}
\end{equation}
the sphere of radius $r$ in Euclidean space, a submanifold of
\begin{equation}
  \Bigl(\mathbb{R}^3,e=(\ud {x^1})^2+(\ud {x^2})^2+(\ud {x^3})^2\Bigr).
\end{equation}
We denote the metric of the round sphere $S_r$ by
\begin{equation}
  \gamma_r=e\bigr\rvert_{\mathrm{T}S_r}=r^2\g\,,
\end{equation}
where $\g$ is the standard metric on the unit sphere $\mathbb{S}^2$.

Let
\begin{equation}
  \label{eq:def:Omegai}
  \Omega_{(i)}=\epsilon_{ijk}\,x^j\,\pd{x^k}\qquad i=1,2,3\,,
\end{equation}
where $\epsilon$ is the volume form of $e$.
We have
\begin{equation}
  \sum_{i=1}^3\Oi{i}^m\Oi{i}^n=\lvert x\rvert^2\delta_{mn}-x^n x^m\,,
\end{equation}
and thus for all $x\in\mathbb{R}^3\backslash\{0\}$ and $u:\mathbb{R}^3\to\mathbb{R}$ differentiable:
\begin{equation}
  \sum_{i=1}^3\bigl(\Oi{i}u\bigr)^2(x)=\lvert x\rvert^2\Bigl[\lvert\nabla u\rvert^2-\bigl\langle\frac{x}{\lvert x\rvert},\nabla u\bigr\rangle^2\Bigr]=\lvert x\rvert^2\lvert\nablab u\rvert^2\,.
\end{equation}
This is the coercivity formula on the sphere. Here $\nablab=\Pi\nabla$, and
\begin{equation}
  \Pi_a^b(\xi)=\delta_a^b-\xi_a\,\xi^b\,,\qquad\xi\in\mathbb{S}^2\,,
\end{equation}
is the projection to the sphere; by uniqueness $\nablab$ is the connection of $\gamma_r$.

\begin{lemma}[Coercivity inequalities on the sphere]
Let $u$ be a smooth function on $S_r$, then
\begin{equation}
  \label{eq:coercivity:1st}
  r^2\lvert\nablab u\rvert_{\gamma_r}^2\leq\sum_{i=1}^3\bigl(\Oi{i}u\bigr)^2\,,
\end{equation}
and
\begin{equation}
  \label{eq:coercivity:2nd}
  r^4\lvert\nablab^2 u\rvert_{\gamma_r}^2\leq\sum_{i,j=1}^3\bigl(\Oi{i}\Oi{j}u\bigr)^2\,.
\end{equation}
\end{lemma}

We have already shown the first inequality. For the second inequality we can use that more generally for any $S_r$-1-form (a 1-form on $\mathbb{R}^3$ such that $\theta\cdot X=\theta\cdot\Pi X$) it holds
\begin{equation}
  \label{eq:lie1form}
  \sum_{i=1}^3\bigl\lvert{\mathcal{L}\!\!\!/}_{\Oi{i}}\theta\bigr\rvert_{\gamma_r}^2=r^2\bigl\lvert\nablab\theta\bigr\rvert_{\gamma_r}^2+\bigl\lvert\theta\bigr\rvert_{\gamma_r}^2\,,
\end{equation}
where $\mathcal{L}\!\!\!/$ denotes the Lie derivative on $S_r$. (To prove \eqref{eq:lie1form} one can proceed analogeously to Lemma 11.2 in \cite{ch:formation}.)
By substituting
\begin{equation}
  \theta=\db u\doteq\ud u\rvert_{\mathrm{T}S_r}\,,
\end{equation}
we then obtain
\begin{multline}\label{eq:coercivity:commute}
  r^2\bigl\lvert\nablab^2 u\bigr\rvert_{\gamma_r}^2=r^2\bigl\lvert\nablab\db u\bigr\rvert^2_{\gamma_r}\leq\\
  \leq\sum_{i=1}^3\bigl\lvert\Lb_{\Oi{i}}\db u\bigr\rvert^2_{\gamma_r}=\sum_{i=1}^3\bigl\lvert\db\Lb_{\Oi{i}}u\bigr\rvert_{\gamma_r}^2=\sum_{i=1}^3\bigl\lvert\nablab(\Oi{i})\bigr\rvert^2_{\gamma_r}\,,
\end{multline}
because Lie derivatives commute with exterior derivatives. Inequality \eqref{eq:coercivity:2nd} then follows from \eqref{eq:coercivity:commute} using \eqref{eq:coercivity:1st}.

We recall the classical Sobolev inequality on the sphere.
\begin{lemma}[Sobolev embedding on $\mathbb{S}^2$]\label{lemma:sobolevembeddingsphere}
 Let $u\in\mathrm{H}^2(\mathbb{S}^2)$, then $u\in\mathrm{L}^\infty(\mathbb{S}^2)$ and
 \begin{equation}
   \label{eq:lemma:sobolev}
   \lVert u\rVert_{\mathrm{L}^\infty(\mathbb{S}^2)}\leq C\,\lVert u\rVert_{\mathrm{H}^2(\mathbb{S}^2)}\,.
 \end{equation}
\end{lemma}

Given a function $u$ on $S_r$ we can apply Lemma \ref{lemma:sobolevembeddingsphere} to $u\circ h$, where
\begin{equation}
  h:\mathbb{S}^2\to S_r,\quad\xi\mapsto r\xi\,.
\end{equation}
Since, in the coordinates
\begin{equation}
  \gamma_r=r^2\g=r^2\g_{AB}\ud y^A\ud y^B\,,
\end{equation}
$h$ is the identity mapping, we get
\begin{equation}
  \lvert\stackrel{\circ}{\nablab}(u\circ h)\rvert^2_{\g}=r^2\lvert(\nablab u)\circ h\rvert_{\gamma_r}^2\,,
  \quad\lvert{\stackrel{\circ}{\nablab}}^2(u\circ h)\rvert^2_{\g}=r^4\lvert(\nablab^2 u)\circ h\rvert^2_{\gamma_r}\,,
\end{equation}
and thus
\begin{multline}
  r\lvert u\rvert\Bigr\rvert_{S_r}\leq C\Bigl(\int_{S_r}\lvert u\rvert^2\dm{\gamma_r}\Bigr)^\frac{1}{2}+C\Bigl(\int_{S_r}r^2\lvert \nablab u\rvert^2\dm{\gamma_r}\Bigr)^\frac{1}{2}\\+C\Bigl(\int_{S_r}r^4\lvert \nablab^2 u\rvert^2\dm{\gamma_r}\Bigr)^\frac{1}{2}\,.
\end{multline}

\begin{cor}\label{cor:sobolev:spherer}
Let $u\in\mathrm{H}^2(S_r)$, then
\begin{multline}
  r\lvert u\rvert\Bigr\rvert_{S_r}\leq C\Bigl(\int_{S_r}\lvert u\rvert^2\dm{\gamma_r}\Bigr)^\frac{1}{2}+C\Bigl(\int_{S_r}\sum_{i=1}^3\bigl( \Oi{i}u\bigr)^2\dm{\gamma_r}\Bigr)^\frac{1}{2}\\+C\Bigl(\int_{S_r}\sum_{i,j=1}^3\bigl(\Oi{i}\Oi{j}u\bigr)^2\dm{\gamma_r}\Bigr)^\frac{1}{2}\,.
\end{multline}
\end{cor}

%% file: kerrdesitter.tex
In this appendix we show the statement of Proposition~\ref{prop:kerr:G} that subextremal Kerr de Sitter spacetimes satisfy the closeness assumptions of Definition~\ref{def:appl:class}.

A revealing aspect of the Boyer-Lindquist coordinates $(t,r,\theta,\phi)$ is that setting $m=0$ in (\ref{eq:kerrdesitter}--\ref{eq:kerrdesitter:h}) reduces the Kerr de Sitter metric $g$ to the \emph{de Sitter}-metric in ``co-rotating'' coordinates. The change of coordinates that transforms the de Sitter metric $g_0$ in standard coordinates $(t_0,r_0,\theta_0,\phi_0)$,
\begin{equation}
  g_0=-\frac{1}{\frac{\Lambda}{3}r_0^2-1}\ud r_0^2+\Bigl(\frac{\Lambda}{3}r_0^2-1\Bigr)\ud t_0^2+r_0^2\bigl(\ud\theta_0^2+\sin^2\theta_0\ud\phi_0^2\bigr)\,,
\end{equation}
to the expression obtained from the Kerr de Sitter metric in Boyer-Lindquist coordinates by setting $m=0$ is well-known since \cite{carter:lecture,hawking:rot:ads} and explicitly given by:
\begin{equation}
  (t_0,r_0,\theta_0,\phi_0)\to (t,r,\theta,\phi)\,,
\end{equation}
where for some fixed $a>0$, with the notation from \eqref{eq:kerrdesitter:notation},
\begin{subequations}\label{eq:trans:rot}
\begin{gather}
  t_0(t)=t\\
  \phi_0(t,\phi)=\phi-\frac{\Lambda}{3}\, a\, t\\
  r_0^2(\theta,r)=\frac{1}{\Delta_0}\Bigl[r^2\Delta_\theta+a^2\sin^2\theta\Bigr]\\
  r_0(\theta,r) \cos\theta_0(\theta,r)= r\cos\theta\,.
\end{gather}
\end{subequations}
Indeed, the de Sitter metric expressed in these ``comoving'' coordinates takes the form:
\begin{equation}
  g_0=f_0+h_0\,,
\end{equation}
where
\begin{equation}
 f_0=\Bigl[\frac{\Lambda}{3}\bigl(r^2+a^2\sin^2\theta\bigr)-1\Bigr]\,\ud t^2-2\frac{\Lambda}{3}\,a\frac{r^2+a^2}{\Delta_0}\sin^2\theta\,\ud t\,\ud \phi+\frac{r^2+a^2}{\Delta_0}\,\sin^2\theta\,\ud\phi^2
\end{equation}
\begin{equation}
   h_0=\frac{\rho^2}{(r^2+a^2)(1-\frac{\Lambda}{3}r^2)}\,\ud r^2+\frac{\rho^2}{\Delta_\theta^2}\,\ud\theta^2\,.
\end{equation}
Thus in comparison to \eqref{eq:kerrdesitter:f} and \eqref{eq:kerrdesitter:h},
\begin{equation}
  g=g_0+\frac{2mr}{\rho^2}\Bigl(\ud t-\frac{a\sin^2\theta}{\Delta_0}\,\ud\phi\Bigr)^2+\frac{2mr\rho^2}{\Delta_r\rvert_{m=0}\Delta_r}\,\ud r^2\,.
\end{equation}
It is then easy to deduce using \eqref{eq:trans:rot} that with respect to $(t_0,r_0,\theta_0,\phi_0)$-coordinates,
\begin{subequations}\label{eq:kerrdesitter:desitter}
\begin{gather}
  \bigl(g-g_0\bigr)_{t_0 t_0}=\frac{2mr}{\rho^2}\frac{\Delta_\theta^2}{\Delta_0^2}\\
  \bigl(g-g_0\bigr)_{t_0 \phi_0}=-\frac{2mr}{\rho^2}\frac{\Delta_\theta}{\Delta_0}\frac{a\sin^2\theta}{\Delta_0}\\
  \bigl(g-g_0\bigr)_{\phi_0\phi_0}=\frac{2mr}{\rho^2}\frac{a^2\sin^4\theta}{\Delta_0^2}\displaybreak[0]\\
  \bigl(g-g_0\bigr)_{r_0r_0}=\frac{2mr}{\rho^2}\frac{1}{\Delta_r}\frac{\bigl(\Delta_0 r_0 (r-r_0)+r^2+a^2\bigr)^2}{(r^2+a^2)(1-\frac{\Lambda r^2}{3})}\\
  \bigl(g-g_0\bigr)_{r_0\theta_0}=-2\frac{2mr}{\rho^2}\frac{1}{\Delta_r}a^2 r_0\sin\theta_0\cos\theta\Bigl(\frac{\Delta_0 r_0 (r-r_0)}{r^2+a^2}+1\Bigr)\\
  \bigr(g-g_0\bigr)_{\theta_0\theta_0}=\frac{2mr}{\rho^2}\frac{1}{\Delta_r}a^4\frac{1-\frac{\Lambda r^2}{3}}{r^2+a^2}r_0^2\sin^2\theta_0\cos^2\theta\,.
\end{gather}
\end{subequations}
Similar albeit slightly longer expressions can be derived for the components of the inverse metric $g^{-1}$ in comparison to $g_0^{-1}$. Note for that purpose that
\begin{equation}
  \label{eq:kerrdesitter:detf}
  \det f=-\frac{\sin^2\theta}{\Delta_0^2}\:\Delta_r\:\Delta_\theta\,,
\end{equation}
and thus
\begin{subequations}
\begin{gather}
  (f^{-1})^{tt}=-\frac{\Delta_0}{\Delta_r\,\Delta_\theta}\biggl[r^2+a^2+\frac{2mr}{\rho^2}\frac{a^2}{\Delta_0}\sin^2\theta\biggr]\\
  (f^{-1})^{t\phi}=-a\frac{\Delta_0}{\Delta_r\Delta_\theta}\biggl[\frac{\Lambda}{3}\Bigl(r^2+a^2\Bigr)+\frac{2mr}{\rho^2}\biggr]\\
  (f^{-1})^{\phi\phi}=-\frac{1}{\sin^2\theta}\frac{\Delta_0^2}{\Delta_r\Delta_\theta}\biggl[\frac{\Lambda}{3}\Bigl(r^2+a^2\sin^2\theta\Bigr)+\frac{2mr}{\rho^2}-1\biggr]
\end{gather}
\end{subequations}
while simply
\begin{equation}
    (h^{-1})^{rr}=\frac{\Delta_r}{\rho^2}\,,\qquad
  (h^{-1})^{\theta\theta}=\frac{\Delta_\theta}{\rho^2}\,.
\end{equation}
Analogous formulas to \eqref{eq:kerrdesitter:desitter} are then found by expressing the components of $g^{-1}$ in $(t_0,r_0,\theta_0,\phi_0)$-coordinates in terms of the components above with respect to $(t,r,\theta,\phi)$-coordinates using \eqref{eq:trans:rot}.

While the above algebraic manipulations essentially reveal the asymptotic closeness property of the Kerr de Sitter family to the de Sitter cosmology, we are here for convenience interested in its precise relation to the Schwarzschild de Sitter spacetime:
\begin{equation}
  g_m=-\frac{1}{\frac{\Lambda r_0^2}{3}+\frac{2m}{r_0}-1}\ud r_0^2+\Bigl(\frac{\Lambda r_0^2}{3}+\frac{2m}{r_0}-1\Bigr)\ud t_0^2+r_0^2\bigl(\ud\theta_0^2+\sin^2\theta_0\ud\phi_0^2\bigr)
\end{equation}
Note that by setting $a=0$ in \eqref{eq:kerrdesitter:desitter} we obtain precisely the difference betweeen $g_m$ and $g_0$ in standard $(t_0,r_0,\theta_0,\phi_0)$-coordinates. Now \eqref{eq:kerrdesitter:desitter} yields as $1/r\to 0$:
\begin{subequations}\label{eq:kerrdesitter:desitter:a}
\begin{gather}
  \bigl(g-g_m\bigr)_{t_0 t_0}=\mathcal{O}\Bigl(\frac{2m}{r_0}\Bigr)\\
  \bigl(g-g_m\bigr)_{t_0 \phi_0}=\mathcal{O}\Bigl(\frac{2m}{r_0}\Bigr)\\
  \bigl(g-g_m\bigr)_{\phi_0\phi_0}=\mathcal{O}\Bigl(\frac{2m}{r_0}\Bigr)\\
  \bigl(g-g_m\bigr)_{r_0r_0}=\frac{1}{\frac{\Lambda}{3}r_0^2+\frac{2m}{r_0}-1}\mathcal{O}\Bigl(\frac{2m}{r_0^3}\Bigr)\displaybreak[0]\\
  \bigl(g-g_m\bigr)_{r_0\theta_0}=\frac{1}{\frac{\Lambda}{3}r_0^2+\frac{2m}{r_0}-1}\mathcal{O}\bigl(\frac{2m}{r_0^2}\bigr)\\
  \bigl(g-g_m\bigr)_{\theta_0\theta_0}=\mathcal{O}\Bigl(\frac{2m}{r_0^3}\Bigr)\,.
\end{gather}
Similarly we note for the components of the inverse:
\end{subequations}
\begin{subequations}\label{eq:kerrdesitter:desitter:inverse:a}
  \begin{gather}
  (g-g_m)^{t_0t_0}=\frac{1}{\frac{\Lambda}{3}r_0^2+\frac{2m}{r_0}-1}\mathcal{O}\Bigl(\frac{a^2}{r_0^2}\Bigr)\\
  (g-g_m)^{t_0\phi_0}=\frac{1}{\frac{\Lambda}{3}r_0^2+\frac{2m}{r_0}-1}\mathcal{O}\Bigl(\frac{2m}{r_0^3}\Bigr)\\
  (g-g_m)^{\phi_0\phi_0}=\mathcal{O}\Bigl(\frac{2m}{r_0^5}\Bigr)\displaybreak[0]\\
  (g-g_m)^{r_0r_0}=\mathcal{O}\Bigl(\frac{2m}{r_0}\Bigr)\\
  (g-g_m)^{r_0\theta_0}=\mathcal{O}\Bigl(\frac{2m}{r_0^4}\Bigr)\\
  (g-g_m)^{\theta_0\theta_0}=\mathcal{O}\Bigl(\frac{2m}{r_0^7}\Bigr)
\end{gather}
\end{subequations}

The condition~(1) of Def.~\ref{def:appl:class} on $\mathrm{C}^0$-convergence to $g_m$ in the chart $(t_0,r_0,\theta_0,\phi_0)$ is thus evidently verified by the Kerr de Sitter family with $\delta=2$.

The results \eqref{eq:kerrdesitter:desitter:a} and \eqref{eq:kerrdesitter:desitter:inverse:a} can also be used to calculate the connections coefficients of $g$ in $(t_0,r_0,\theta_0,\phi_0)$-coordinates, which in turn lead to detailed bounds on the components of ${}^\prime\pi$, c.f.~Def.~\ref{def:appl:class}~(2).
For definiteness we note here:
\begin{subequations}
\begin{gather}
  {}^\prime\pi_{t_0t_0}-{({}^\prime\pi_m)}_{t_0t_0}
  =\Bigl(\frac{\Lambda r_0^2}{3}+\frac{2m}{r_0}-1\Bigr)\mathcal{O}\Bigl(\frac{1}{r_0^3}\Bigr)\\
  {}^\prime\pi_{t_0\phi_0}
=\mathcal{O}\Bigl(\frac{2m}{r_0^2}\Bigr)\displaybreak[0]\\
  {}^\prime\pi_{t_0r_0}=0\qquad {}^\prime\pi_{t_0\theta_0}=0\\
  {}^\prime\pi_{\phi_0\phi_0}-{({}^\prime\pi_m)}_{\phi_0\phi_0}
  =\mathcal{O}\Bigl(\frac{2m}{r_0^2}\Bigr)\displaybreak[0]\\
    {}^\prime\pi_{\phi_0 r_0}=0\qquad {}^\prime\pi_{\phi_0\theta_0}=0\qquad
  {}^\prime\pi_{r_0t_0}=0\qquad{}^\prime\pi_{r_0\phi_0}=0\\
    {}^\prime\pi_{r_0r_0}-{({}^\prime\pi_m)}_{r_0r_0}
  =\frac{1}{\frac{\Lambda r_0^2}{3}+\frac{2m}{r_0}-1}\mathcal{O}\Bigl(\frac{2m}{r_0^4}\Bigr)\displaybreak[0]\\
    {}^\prime\pi_{r_0\theta_0}
  =\frac{1}{\frac{\Lambda}{3}r_0^2+\frac{2m}{r_0}-1}\mathcal{O}\bigl(\frac{2m}{r_0^3}\bigr)\\
    {}^\prime\pi_{\theta_0\theta_0}-{({}^\prime\pi_m)}_{\theta_0\theta_0}
  =\mathcal{O}\Bigl(\frac{2m}{r_0^4}\Bigr)
\end{gather}
\end{subequations}
This completes the proof of the asymptotic closeness properties for Proposition~\ref{prop:kerr:G}.

%% file: glw.bbl
\providecommand{\bysame}{\leavevmode\hbox to3em{\hrulefill}\thinspace}
\providecommand{\MR}{\relax\ifhmode\unskip\space\fi MR }
\providecommand{\MRhref}[2]{%
  \href{http://www.ams.org/mathscinet-getitem?mr=#1}{#2}
}
\providecommand{\href}[2]{#2}